\documentclass[letterpaper,11pt]{amsart}
\usepackage{latexsym,array,blkarray,delarray,amsthm,amssymb,epsfig,amsmath,blkarray,tikz,algorithmic}
\usepackage{url}
\usepackage{tkz-euclide}
\usepackage{subfig}
\usepackage{subcaption}
\usepackage{circuitikz}
\usepackage{longtable}
\usetikzlibrary{decorations.markings}
\usetikzlibrary{arrows}

\usepackage[dvipsnames]{xcolor}

\theoremstyle{plain}
\newtheorem{thm}{Theorem}[section]
\newtheorem{lemma}[thm]{Lemma}

\newtheorem*{thm*}{Theorem}
\newtheorem*{lemma*}{Lemma}
\newtheorem*{prop*}{Proposition}
\newtheorem*{cor*}{Corollary}
\newtheorem*{conj*}{Conjecture}
\newtheorem*{rep@theorem}{\rep@title}
\newcommand{\newreptheorem}[2]{%
\newenvironment{rep#1}[1]{%
 \def\rep@title{#2 \ref{##1}}%
 \begin{rep@theorem}}%
 {\end{rep@theorem}}}
 \newreptheorem{theorem}{Theorem}

\theoremstyle{definition}
\newtheorem{defn}[thm]{Definition}
\newtheorem{ex}[thm]{Example}

\theoremstyle{remark}

\newcommand{\pp}{\mathbb{P}}

\newcommand{\rr}{\mathbb{R}}

\newcommand{\ind}{\mbox{$\perp \kern-5.5pt \perp$}}

\newcommand{\rank}{\mathrm{rank}}

\newcommand{\trunc}{\mathrm{trunc}}
\newcommand{\pa}{\mathrm{pa}}
\newcommand{\nd}{\mathrm{nd}}
\newcommand{\an}{\mathrm{an}}
\newcommand{\de}{\mathrm{de}}
\newcommand{\ch}{\mathrm{ch}}
\newcommand{\term}{\mathrm{term}}
\newcommand{\pinch}{\mathrm{pinch}}
\newcommand{\diag}{\mathrm{diag}}
\newcommand{\Flat}{\mathrm{Flat}}
\newcommand{\sib}{\mathrm{sib}}

\newcommand{\indep}{\perp \!\!\! \perp}

\begin{document}
\title{Identifiability of a Simple Model of Lateral Gene Transfer}
\author{Devon Olds, Seth Sullivant}
\address{Department of Mathematics \\
North Carolina State University, Raleigh, NC, 27695}
\email{dolds@ncsu.edu}
\email{smsulli2@ncsu.edu}

\maketitle
\begin{abstract}
In evolutionary biology, factors like lateral gene transfer complicate the tree of life, making inference of a species tree more difficult. 
The phenomenon of lateral gene transfer allows genetic material to be passed between organisms as opposed to inheritance, causing the tree for a particular gene to differ from the species tree. 
In this work, we define a model of lateral gene transfer
on site patterns.
In the case where lateral gene transfer is restricted to occur only between
closely related species, we show that the unrooted topology of the species tree is identifiable from SNP data on the taxa.
Our proof involves showing a connection between our lateral gene transfer
model and graphical models on a related tree, and uses ranks of flattenings
to identify splits in the tree.  We also report on results of using
the singular value decomposition on flattening matrices to identify
the unrooted topology in simulated data.  
\end{abstract}


\section{Introduction}

In the study of evolutionary biology, a common goal is to determine the true evolutionary history of a collection of  species based on the 
data about those species. 
The evolutionary history is typically represented by a (rooted) phylogenetic
tree, the tree-structure showing the history of speciations that lead
to the extant species, represented by the leaves of the tree.  

The most basic phylogenetic models are simple Markov processes
on a tree, and the observed data is DNA sequence data from a single gene, common to all
the species.
Non-tree-like processes are also present in evolution,
and it has become increasingly important to study and model these
structures to recover the correct evolutionary structure, and to have more complex data sets. Lateral gene transfer, hybridization, recombination, and gene duplication are all examples of processes which cannot be described by a purely tree-like model of evolution \cite{reticulate_networks_based_on_coalescent,consensus_networks,application_of_phylo_networks,survey_combinatorial_methods_phylo_networks}.  

In this paper, we consider a model for the evolution of DNA on a tree with the addition of lateral gene transfer. Lateral, or horizontal, gene transfer is the process by which genetic material is transferred directly between organisms (horizontally) as opposed to inheritance (vertically). Lateral gene transfer is especially common among prokaryotes, and may play a critical role in the development of antibiotic resistance \cite{lgt_eukaryotes,Barlow2009-ez,Garcia-Vallve2000-ow}. Additionally, lateral gene transfer has been recorded in eukaryotes, and even between kingdoms (for example, between plants and fungi) \cite{Graham2021-mn,Quispe-Huamanquispe2017-ae}. This makes phylogenetic inference complicated, as the history of an individual gene can be quite adventurous compared to the history of the species as a whole. Tools to determine phylogenetic trees despite these aberrant gene lineages are becoming increasingly necessary as we uncover the extent of lateral gene transfer in nature. Alternatively, some research has found that lateral gene transfer events can provide information themselves about the history of a species. Research in \cite{Abby2012-vt} and \cite{Szollosi2012-tp} found that lateral gene transfer can help distinguish clades of prokaryotic life, which helps establish a timeline for a branch of life with a limited fossil record. These studies further show that the study of lateral gene transfer is fruitful and worthy of consideration in biological applications.

Since the gene tree and species tree may differ for a particular gene and species in which lateral gene transfer occurs, our main goal is to infer the species tree from data on the genes. Although lateral gene transfer can obscure the original species tree, it has been shown that the species tree is identifiable with surprisingly high amounts of gene transfer 
\cite{daskalakis2017speciestreesrecoverableunrooted,Roch2013-tc}.
However, these results infer
the species tree from observations of gene trees.  A difference with our
approach is that we are considering SNP data rather than gene trees as
data. SNP data, or single nucleotide polymorphism data, describes the variation of a single gene among a population of organisms. Since lateral gene transfer affects large contiguous sections of the genome in a single event, models that assume site independence of mutations are only appropriate when considering SNPs \cite{utility_snps,Phylogeographic_reconstructionq_high_lgt,Phylogenetic_understanding__clonal_populations}. 

The framework of investigating probabilities of site patterns using SNP data is also used in \cite{kubatko2026revisitingrandommodellateral}, where they examine identifiability of species trees and tree parameters in the two and three taxon cases. In the both cases they are able to directly compute the site pattern probabilities by integrating over all possible gene trees and the densities of the divergence times. This is made possible by the fact that the gene tree is determined by a finite number of relevant gene transfers. In the two taxon case, they also are able to show that the lateral gene transfer rate and the time of divergence in the species tree are each generically identifiable given the other. In the three taxon case, they show that the species tree topology is also generically identifiable.

Another unique element of our approach is to restructure the Markov process to incorporate lateral gene transfer. Basic models of evolution on a tree use a Markov process across the edges \cite{allman2006phylogeneticidealsvarietiesgeneral,leaf_colourations}. This works well for independently evolving lineages but has trouble representing lateral gene transfer events which span multiple lineages. 
When such an event takes place, the state of the gene on one edge overwrites the state of the gene on a second edge. In this way, random variables potentially depend on both their ancestor variables and  concurrent, non-ancestor edges. Basic Markov models do not account for both the standard evolution process and gene transfer. We will define a new Markov process which describes the evolution of sequences of states across all contemporaneous edges. Since multiple lineages will be described with the same transition matrix, the Markov process will be able to include lateral gene transfer events.

To understand and describe this Markov process, we first give background on Markov models on
trees in Section \ref{sec:Markov}, including variations on trees structures we will need for our
model.  
In Section \ref{sec:local} we present a general version of our
lateral gene transfer model.  Then, we focus on a local version of the model,
which only allows lateral gene transfer events to occur between closely
related species.  In Section \ref{sec:graphical}, we explain how
to interpret the local model within the context of the family of graphical models.
Using this perspective, we prove the main result of the paper:
\begin{thm}
    Let $T$ and $T'$ be tiered X-trees with differing unrooted topologies. Then under the Local Gene Transfer Model these trees are generically distinguishable.
\end{thm}
The proof depends on using the connection to graphical models to show that
flattening matrices associated to splits in the graph have low rank,
following a similar approach from many articles in algebraic approaches
to phylogenetics \cite{Tree_Symmetries, DBLP:journals/jcb/AllmanR06, allman2006phylogeneticidealsvarietiesgeneral, Chifman2014-kl}. 

The fact that flattening matrices associated to splits in the tree have low
rank leads to an approach for reconstructing trees using the singular value decomposition \cite{Chifman2014-kl}.  In Section \ref{sec:sims}, we test this approach 
on the local lateral gene transfer model on simulated data.  Somewhat surprisingly,
our simulations show that in the local model, increased amounts of 
lateral gene transfer increases the probability that the SVD-based approach
infers the correct underlying tree.


\section{Markov Models on Trees}\label{sec:Markov}


This section contains background on trees and Markov models on trees.
We will specify what it means to be a phylogenetic $X$-tree, 
and how evolutionary data is represented by this tree. We will use this information to describe the relationships between the taxa in the form of splits, and subsequently describe the structure of the tree.  One component that will be useful for Section \ref{sec:local}
is a non-standard description of the Markov process on an equidistant tree, via a Markov process on regions of the tree we call ``tiers''.    More detailed 
background information on phylogenetic trees appears in \cite{Phylogeny}.

\begin{defn}
    A \textit{directed graph} is a pair of sets $G=(V,E)$, with vertices $V$ and the edges a set of ordered pairs $E\subseteq {V\choose 2}$. A \textit{rooted binary tree} is a graph containing no cycles,  with a vertex $\rho$ called the ``root," with every edge directed away from $\rho$, and with every interior vertex having in-degree one and out-degree two (except for $\rho$, which has in-degree zero and out-degree two). The \textit{leaves} of the tree are the exterior vertices having in-degree one and out-degree zero.
\end{defn}

An \textit{unrooted} binary tree is similar to a rooted binary tree, 
but with undirected edges, no root, and all interior vertices have degree three while exterior vertices have degree one. A rooted tree can be converted 
to an unrooted tree by suppressing the root vertex. The root $\rho$ is deleted and the two edges incident with $\rho$, $(\rho \rightarrow x)$ and $(\rho \rightarrow y)$ are replaced with a single undirected edge, $(x,y)$.

The vertices of a rooted tree can be described by their relationship with the other vertices. More generally, this also holds for directed acyclic graphs.

\begin{defn}
        Let $G=(V,E)$ be a directed acyclic graph. The set of \textit{descendants} of a vertex $v$ is the set $\de(v)$ consisting of all $u$ such that there is a directed path from $v$ to $u$. The \textit{non-descendants} of $v$, denoted $\nd(v)$ consists of all vertices $w$ such that $w\notin \de(v)$.
    Similarly, the set of \textit{ancestors} of a vertex $u$ is the set $\an(v)$ consisting of all $v$ such that there is a directed path from $v$ to $u$.
    The set of \textit{parents} of $v$, denoted $\pa(v)$, consists of all vertices $u$ such that $u\rightarrow v$ is an edge of $G$. The set of \textit{children} of $u$, denoted $\ch(u)$, consists of all vertices $v$ such that $u\rightarrow v$ is an edge of $G$.
\end{defn}

A tree can be represented by its leaves and the relationships between disjoint sets of its leaves. Each edge of the tree represents a way to separate one set of leaves from another. This gives intuition on which taxa may be more closely related than others.

\begin{defn}
    A \textit{rooted phylogenetic X-tree} is a rooted tree $T$ such that $X$ is the set of leaves.
\end{defn}

\begin{defn}
    A \emph{split} $A|B$ of the set $X$ is a bipartition of $X$ into two nonempty subsets, $A$ and $B$.
    Given an $X$-tree $T$, if we delete a particular edge $e$, then we obtain an induced partition of $X$. We will refer to this split as the \textit{split of T} corresponding to the edge $e$. 
    An $X$-split obtained from an edge this way is also called a \emph{valid} split.   
    Let $\Sigma(T)$ denote the set of all valid splits of $T$. 
\end{defn}

  Note that all phylogenetic $X$-trees have the \textit{trivial splits} $\{x\}|X-\{x\}$ 
  for each element $x\in X$, obtained by deleting the edge incident with the leaf $x$.

A pair of $X$ splits $A|\bar{A}$ and $B|\bar{B}$ is called \emph{compatible} 
if at least one of the following intersections is empty: $A\cap B, A\cap \bar{B},\bar{A}\cap B,$ and $ \bar{A}\cap \bar{B} $.
A set $\Sigma$  of $X$-splits is 
 \textit{pairwise compatible}   if every pair of splits $A|\bar{A}, B|\bar{B} \in \Sigma$ is compatible.  Pairwise compatibility is the key property
 that determines if a set of splits comes from a tree.  

\begin{thm}\cite{Buneman_Splits,Phylogeny}
    \label{splits_theorem}
    Let $\Sigma$ be a set of $X$-splits.  There exists an unrooted phylogenetic $X$-tree
    such that $\Sigma = \Sigma(T)$ if and only if $\Sigma$ is
    pairwise compatible and contains the trivial splits.
    Furthermore, the tree $T$ is uniquely determined by $\Sigma(T)$; that is,
    if $\Sigma(T) = \Sigma(T')$ then $T$ and $T'$ are equivalent.  
\end{thm}

Being able to identify the unrooted structure of a tree through $X$-splits is helpful when the data associated to the leaves allows us to group them into splits. This idea will be
used in Section \ref{sec:graphical} to prove that the unrooted tree structure is identifiable in the lateral gene transfer model.

It will be useful for us to assign lengths to the edges of the trees we work with. Edge lengths can be representative of time, or evolutionary distance, and are useful in many applications of phylogenetic trees. 
Let $l:  E \rightarrow \mathbb{R}_{> 0 }$ be a function that assigns a positive real number to each edge $e$. 
Define $d=d_{(T,l)}\colon X\times X \rightarrow \mathbb{R}$ by letting $d(x,y)$ be the sum of the edge lengths on the unique path between the vertices $x$ and $y$. 
Since $d$ satisfies the triangle inequality, $d$ is a metric, and since $d$ is represented by a tree, we will call it a \textit{tree metric}.

\begin{defn}
    An \textit{ultrametric} is any metric $d$ on $X$ which satisfies the property that for any three pairwise distinct elements 
    $x,y,z\in X$, $d(x,y)\leq \max (d(x,z),d(y,z))$.
\end{defn}

Ultrametrics are commonly used in the rooted setting by assigning lengths $l(e)>0$ to each edge in the tree so that the distance from the root $\rho$ to each leaf is the same.   Such an assignment is called \emph{equidistant}.  
This can be especially useful if the edge lengths represent time, and the same amount of time has passed from the ancestral species to the present-day species.
Throughout the rest of the paper, we will generally assume that 
edge lengths on a tree are ultrametrics.  This requirement will be an 
important simplifying component in our description of the lateral gene
transfer model in the next section.

Now that we have a framework for measuring time across a tree, 
we can begin to model how genes evolve over time. 
We will use a continuous Markov process to 
describe the process of sites changing over time in diverging lineages. 

\begin{defn}
    Let $Y_0,Y_1, \ldots ,Y_m$ be a sequence of random variables, each with values from a set of states $S$. 
    We say this sequence is a \textit{nonhomogeneous Markov chain} if for each $i>0$, $Y_{i+1}$ is conditionally independent from $Y_0, \ldots,Y_{i-1}$ given $Y_i$. This means that each random variable only relies on the variable just before, and disregards all other past variables. 
\end{defn}

To model the process of changing genes over time, we use a continuous time 
Markov process.
Let $\pi$ be a vector of length $r=|S|$ where $\pi_\alpha  =  \pp(Y_0 = \alpha)$. 
Let $P^{(i)}$ be an $r\times r$ \textit{ transition matrix} of conditional probabilities 
\[   P^{(i)}_{\alpha \beta} := \pp (Y_{i+1}=\beta |Y_i= \alpha).   \]
We will assume that the matrix $P$ and the vector $\pi$ are both strictly positive.  Also note that the rows of $P$ sum to one. 
 A \textit{rate matrix} is an $r \times r $ matrix $Q$ which has nonnegative off-diagonal entries and rows summing to zero, so the diagonal entries must be the negative sum of the rest of the row. The entry $Q_{\alpha \beta}$ represents the rate at which a site transitions from state $\alpha$ to state $\beta$.
 When we use rate matrices in a phylogenetic model, the matrix $Q$ may be different for each edge, and each $Q$ is in effect for some nonnegative amount of time $t\geq 0$.  
The transition matrix $P(t)$ is then the solution to the differential equation $dP(t)/dt = P(t)Q$ with the initial condition $P(0)=I$:
\[P(t) = \exp(Qt).\]
For a simple Markov chain, the transition matrix for each succeeding edge can be multiplied to obtain a transition matrix from the first variable to the last, across all edges in between. Summing over each of the possible states of the intermediary variables and assuming that the first variable $Y_0 = y_0$,
\[   
\pp(Y_m=y_m)=  \sum_{y_1\in S}\cdots \sum_{y_m \in S} \pp(Y_0=y_0) \prod_{1\leq i \leq m} \pp(Y_i=y_i |Y_{i-1} =y_{i-1} ).     
\]

For a tree with leaf set $X$, we sum over the possible states of all internal vertices, and multiply the probability of each state given the state of its parent vertex.  Recall that $\pa(i)$ denotes the parent of the vertex $i$.  For a tree with root $Y_0$, leaf set $\{Y_{m+1},...,Y_n\}$, and internal vertices 
$\{Y_1,\ldots ,Y_m\}$,

\[  \pp(Y_{m+1}=y_{m+1}, \ldots , Y_n=y_n)=  \sum_{y_0\in S}\cdots \sum_{y_m \in S} \pp(Y_0=y_0) \prod_{1\leq i \leq m} \pp(Y_i=y_i |Y_{\pa(i)} = y_{\pa(i)}) .  \]

\begin{ex}
   Let $T$ be the tree on four leaves as shown in Figure \ref{fig: ex no lgt}. 
   Each edge of the tree has an associated transition matrix $P_i$, and each vertex $i$ has an associated variable $Y_i$ with state $y_i$.

    \begin{align*}
        &\pp(Y_3=y_3, Y_4=y_4,Y_5=y_5 , Y_6=y_6)\\ &=  \sum_{y_0\in S}\sum_{y_1 \in S}\sum_{y_3\in S} \pp(Y_0=y_0) \pp(Y_1=y_1|Y_0=y_0)\pp(Y_2=y_2|Y_1=y_1)\pp(Y_3=y_3|Y_1=y_1)\\
        & \hspace{46mm}   \pp(Y_4=y_4|Y_2=y_2)\pp(Y_5=y_5|Y_2=y_2)\pp(Y_6=y_6|Y_0=y_0)
    \end{align*}

\end{ex}   
    
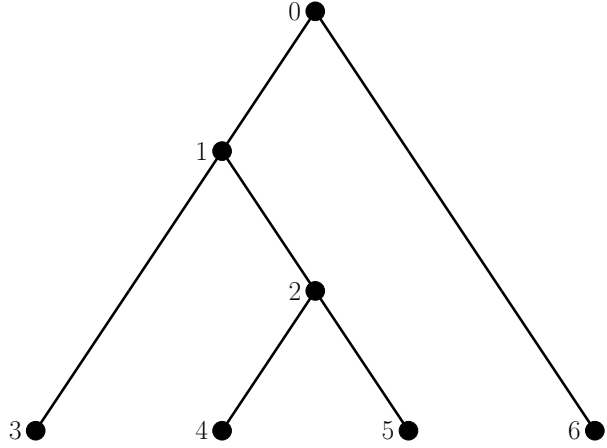
\begin{figure}[]
\centering
\resizebox{.5\textwidth}{!}{%
\begin{circuitikz}
\tikzstyle{every node}=[font=\huge]

\draw [ fill={rgb,255:red,0; green,0; blue,0} ] (-2.5,13.25) circle (0.25cm) node [left=7pt]{$0$};

\draw [line width=2pt, short] (-2.5,13.25) -- (-5,9.5);
\draw [line width=2pt, short] (-2.5,13.25) -- (0,9.5) ;

\draw [ fill={rgb,255:red,0; green,0; blue,0} ] (-5,9.5) circle (0.25cm) node [left=7pt]{$1$};
\draw [ fill={rgb,255:red,0; green,0; blue,0} , line width=0.2pt ] (-2.5,5.75) circle (0.25cm) node [left=7pt]{$2$};
\draw [line width=2pt, short] (-5,9.5) -- (-7.5,5.75);
\draw [line width=2pt, short] (-5,9.5) -- (-2.5,5.75);

\draw [line width=2pt, short] (-2.5,5.75) -- (-5,2);
\draw [line width=2pt, short] (-2.5,5.75) -- (0,2) ;

\node [font=\LARGE] at (-3.25,3.25) {};
\draw [line width=2pt, short] (0,9.5) -- (2.5,5.75);
\draw [line width=2pt, short] (2.5,5.75) -- (5,2);
\draw [line width=2pt, short] (-7.5,5.75) -- (-10,2);
\draw [ fill={rgb,255:red,0; green,0; blue,0} , line width=0.2pt ] (-10,2) circle (0.25cm) node [left=7pt]{$3$};
\draw [ fill={rgb,255:red,0; green,0; blue,0} , line width=0.2pt ] (0,2) circle (0.25cm) node [left=7pt]{$5$};
\draw [ fill={rgb,255:red,0; green,0; blue,0} , line width=0.2pt ] (-5,2) circle (0.25cm) node [left=7pt]{$4$};
\draw [ fill={rgb,255:red,0; green,0; blue,0} , line width=0.2pt ] (5,2) circle (0.25cm) node [left=7pt]{$6$};

\end{circuitikz}
}%
\caption{A rooted tree on four leaves.}
\label{fig: ex no lgt}
\end{figure}

We can also think about grouping the lineages together, so that instead of having an individual matrix for each edge of the tree, we can use a single matrix to represent all lineages for a certain length of time. Then, we can multiply these matrices in chronological order to obtain the probabilities of all leaves. To do this, we first need to add extra nodes to our tree. For every vertex representing a divergence of lineages,  we will add nodes on every other contemporaneous edge of the tree so that all edges which exist at the same time have the same length, as in Figure \ref{fig:levels}.

\begin{ex}
Consider the tree $T$ in Figure \ref{fig:levels}.  Each edge of the tree has an associated transition matrix $P_i$, and each vertex $i$ has an associated variable $Y_i$ with state $y_i$.
\begin{align*}
        &\pp(Y_6=y_6, Y_7=y_7,Y_8=y_8 , Y_9=y_9)\\ &=  \sum_{y_0\in S}\sum_{y_1 \in S}\sum_{y_2\in S}\sum_{y_3\in S}\sum_{y_4\in S}\sum_{y_5\in S} \pp(Y_0=y_0) \pp(Y_1=y_1|Y_0=y_0)\pp(Y_2=y_2|Y_0=y_0)\pp(Y_3=y_3|Y_1=y_1)\\
        & \hspace{68mm}   \pp(Y_4=y_4|Y_1=y_1)\pp(Y_5=y_5|Y_2=y_2)\pp(Y_6=y_6|Y_3=y_3)\\
        & \hspace{68mm}   \pp(Y_7=y_7|Y_4=y_4)\pp(Y_8=y_8|Y_4=y_4)\pp(Y_9=y_9|Y_5=y_5)
    \end{align*}

\end{ex}
    
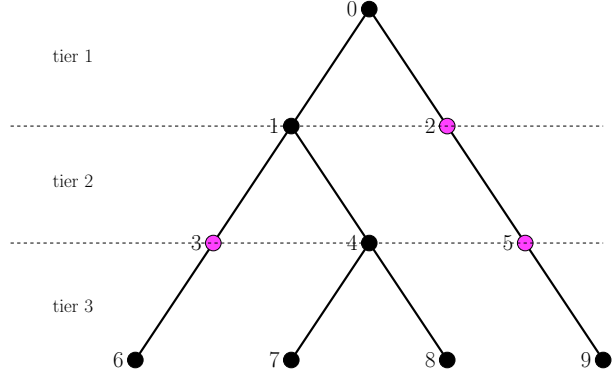
\begin{figure}[]
\centering
\resizebox{.5\textwidth}{!}{%
\begin{circuitikz}
\tikzstyle{every node}=[font=\huge]

\draw [ fill={rgb,255:red,0; green,0; blue,0} ] (-2.5,13.25) circle (0.25cm) node [left=7pt]{$0$};

\draw [line width=2pt, short] (-2.5,13.25) -- (-5,9.5);
\draw [line width=2pt, short] (-2.5,13.25) -- (0,9.5) ;

\draw [ fill={rgb,255:red,0; green,0; blue,0} ] (-5,9.5) circle (0.25cm) node [left=7pt]{$1$};
\draw [ fill={rgb,255:red,0; green,0; blue,0} , line width=0.2pt ] (-2.5,5.75) circle (0.25cm) node [left=7pt]{$4$};

\draw [line width=2pt, short] (-5,9.5) -- (-7.5,5.75);
\draw [line width=2pt, short] (-5,9.5) -- (-2.5,5.75);

\draw [line width=2pt, short] (-2.5,5.75) -- (-5,2);
\draw [line width=2pt, short] (-2.5,5.75) -- (0,2);

\draw [line width=2pt, short] (0,9.5) -- (2.5,5.75);
\draw [line width=2pt, short] (-7.5,5.75) -- (-10,2);
\draw [line width=2pt, short] (2.5,5.75) -- (5,2);

\draw [ fill={rgb,255:red,255; green,64; blue,255} , line width=0.2pt ] (0,9.5) circle (0.25cm) node [left=7pt]{$2$};
\draw [ fill={rgb,255:red,255; green,64; blue,255} , line width=0.2pt ] (-7.5,5.75) circle (0.25cm) node [left=7pt]{$3$};
\draw [ fill={rgb,255:red,255; green,64; blue,255} , line width=0.2pt ] (2.5,5.75) circle (0.25cm) node [left=7pt]{$5$};

\draw [ fill={rgb,255:red,0; green,0; blue,0} , line width=0.2pt ] (-10,2) circle (0.25cm) node [left=7pt]{$6$};
\draw [ fill={rgb,255:red,0; green,0; blue,0} , line width=0.2pt ] (0,2) circle (0.25cm) node [left=7pt]{$8$};
\draw [ fill={rgb,255:red,0; green,0; blue,0} , line width=0.2pt ] (-5,2) circle (0.25cm) node [left=7pt]{$7$};
\draw [ fill={rgb,255:red,0; green,0; blue,0} , line width=0.2pt ] (5,2) circle (0.25cm) node [left=7pt]{$9$};

\draw [dashed] (-14,9.5) -- (5,9.5);
\draw [dashed] (-14,5.75) -- (5,5.75);
\node [font=\LARGE] at (-12,11.75) {tier 1};
\node [font=\LARGE] at (-12,7.75) {tier 2};
\node [font=\LARGE] at (-12,3.75) {tier 3};

\end{circuitikz}
}%
\caption{The pink nodes $Y_2,Y_3,$ and $Y_5$ are the added nodes.  The dashed lines mark the boundary between tiers.}
\label{fig:levels}
\end{figure}

After adding nodes in this way, we can think of the tree being separated into different tiers, with each edge in tier $i$ having length $t_i$.

\begin{defn}
    A \textit{tiered X-tree} $T$ is a rooted binary phylogenetic $X$-tree such that
    if we assign a length of $1$ to every edge, then the resulting tree metric
    is an ultrametric.  
    A \textit{tier} in $T$ is a set of vertices all with the same distance $\delta$ from the root along with the edges incident with and directed toward these vertices.
\end{defn}

To construct a tiered $X$-tree from a given rooted tree, we can add extra
nodes of degree $2$ that subdivide the edges.  Note that if we have a tiered $X$-tree
and an ultrametric $d$ associated to that tree, we will assume that the edge lengths
in each tier are all the same.  

To achieve our goal of having a unifying $Q$ matrix for each tier, we will create a rate matrix describing the transition between sequences of states rather than individual states. Let $q^{(a)}$ be the transition-state matrix for the edge representing species $a$. Let $\mathrm{d}_H(i,j)$ be the Hamming distance between the sequences $i = (i_1, \ldots, i_n) $ and $j = (j_1, \ldots, j_n)$, where $i, j \in S^n$.

\[     Q(i;j) = \begin{cases}
    q^{(a)}_{i_a, j_a} &\mathrm{d}_H(i,j)=1, i_a \neq j_a \\
    -\sum_{k \in S^n:k \neq i }{Q(i;k)} & \mathrm{d}_H(i,j)=0\\
    0 &\mathrm{d}_H(i,j)>1
\end{cases}   \]
Equivalently, we can construct transition-state matrices as a Kronecker sum of matrices corresponding to each edge. 

Let $A$ be an $m \times m$ matrix, let $B$ be an $n \times n$ matrix, and let
$I_m$ and $I_n$ denote $m \times m$ and $n \times n$ identity matrices, respectively.
Let $\otimes$ denote the Kronecker product, 
and let $\oplus$ denote the Kronecker sum, defined by $A\oplus B = (A\otimes I_n)+(I_m \otimes B)$; that is, 
 \begin{align*}
    (A\oplus B)_{(nr+v,ns+w)}&=(A\otimes I_n)_{(nr+v,ns+w)}+(I_m\otimes B)_{(nr+v,ns+w)}\\
    &= A_{(r,s)}(I_{n})_{(v,w)}+(I_{m})_{(r,s)}B_{(v,w)}.
\end{align*}     

\begin{thm}
    Let $T$ be a tiered $X$-tree. The Kronecker sum of the rate matrices for each edge within a single tier of $T$ is itself a valid rate matrix, representing the changes between sequences of states.
\end{thm}

\begin{proof}
    We will prove that the Kronecker sum of two rate matrices of arbitrary size satisfies the requirements of a rate matrix.  The general result follows by induction.  Let $A$ be an arbitrary $m\times m$ matrix and $B$ an arbitrary  $n\times n$ matrix, both satisfying the conditions of $Q$.

    Relabeling the indices as sequences:
\begin{align*}
   & A_{(r,s)}(I_{n})_{(v,w)}+(I_{m})_{(r,s)}B_{(v,w)} \\
   = &A_{(r_1,...,r_m,s_1,...,s_m)}(I_{n})_{(v_1,...,v_n,w_1,...,w_n)}+(I_{m})_{(r_1,...,r_m,s_1,...,s_m)}B_{(v_1,...,v_n,w_1,...,w_n)}\\
   = &A_{(r_1,...,r_m,s_1,...,s_m)}\delta_{((v_1,...,v_n),(w_1,...,w_n))}+\delta_{((r_1,...,r_m),(s_1,...,s_m))}B_{(v_1,...,v_n,w_1,...,w_n)}
\end{align*}
where $\delta_{i,j}$ is the Kronecker delta.

If $(v_1,...,v_n)=(w_1,...,w_n) $ and $(r_1,...,r_m)\neq (s_1,...,s_m)$, then in the combined sequence transitioning from $(r_1,...,r_m,v_1,...,v_n)$ to $(s_1,...,s_m,w_1,...,w_n)$ there is exactly one state change, and $(A\oplus B)_{nr+v,ns+w}=A_{r,s}$. Similarly, if $(v_1,...,v_n)\neq(w_1,...,w_n) $ and $(r_1,...,r_m)=(s_1,...,s_m)$, then $(A\oplus B)_{nr+v,ns+w}=B_{v,w}$.  If $(v_1,...,v_n)\neq(w_1,...,w_n) $ and $(r_1,...,r_m)\neq(s_1,...,s_m)$, then there are at least two state changes, and thus $(A\oplus B)_{nr+v,ns+w}=0$. If $(v_1,...,v_n)=(w_1,...,w_n) $ and $(r_1,...,r_m)=(s_1,...,s_m)$,  we have that
\begin{align*}
    (A\oplus B)_{nr+v,nr+v} &= -\sum_{j:j\neq nr+v}{(A\oplus B)_{nr+v,j}}\\
    &= -\sum_{j:j\neq nr+v}{(A\otimes I_n)_{nr+v,j}+(I_m\otimes B)_{nr+v,j}}\\
    &=-\sum_{j:j\neq r}{A_{r,j}}-\sum_{j:j\neq v}{B_{v,j}}\\
    &=A_{r,r}+B_{v,v}
\end{align*}
 Thus, $(A\oplus B)$ also satisfies the conditions of being a $Q$ matrix, and accurately represents the structure of the tree. Since the Kronecker sum is associative, the sum of two matrices can be extended to the sum of arbitrarily many matrices. Therefore the Kronecker sum of every rate matrix within a tier is also a valid rate matrix.
\end{proof}

\begin{ex} \label{Q1}
    Consider the tiered tree $T$ in Figure \ref{fig:levels}. The first tier has edges with transition matrices $P_1$ and $P_2$. Suppose the corresponding rate matrices are 
    \[
    q_1=\begin{bmatrix}
        -a&a\\b&-b
    \end{bmatrix},
   q_2= \begin{bmatrix}
        -c&c\\d&-d
    \end{bmatrix}.
    \]
    Then the rate matrix for the tier as a whole is 
    \begin{align*}
    Q_1 &= \begin{bmatrix}
        -a&a\\b&-b
    \end{bmatrix} \oplus \begin{bmatrix}
        -c&c\\d&-d
    \end{bmatrix}\\
    &= \begin{bmatrix}
        -a&a\\b&-b
    \end{bmatrix}\otimes \begin{bmatrix}
        1&0\\0&1
    \end{bmatrix} + \begin{bmatrix}
        1&0\\0&1
    \end{bmatrix}\otimes \begin{bmatrix}
        -c&c\\d&-d
    \end{bmatrix}\\
    &= \begin{bmatrix}
        -a&0&a&0\\0&-a&0&a\\b&0&-b&0\\0&b&0&-b
    \end{bmatrix}+\begin{bmatrix}
        -c&c&0&0\\d&-d&0&0\\0&0&-c&c\\0&0&d&-d
    \end{bmatrix}\\
    &=\begin{bmatrix}
        -a-c&c&a&0\\d&-a-d&0&a\\b&0&-b-c&c\\0&b&d&-b-d
    \end{bmatrix}
    \end{align*}

When we assign sequences to the rows and columns, we can see that $a$ and $b$ represent the
infinitesimal rate of change in the first position of the sequences, 
$c$ and $d$ represent the infinitesimal rate of change in the second position, and there is a zero in the entries corresponding to changes in both position.

        \[ Q_1=
\begin{blockarray}{ccccc}
        & 00 & 01 & 10 & 11 \\
      \begin{block}{c[cccc]}
        00 & -a-c & c & a & 0 \\
        01 & d & -a-d & 0 & a \\
        10 & b & 0 & -b-c & c \\
        11 & 0 & b & d & -b-d \\
      \end{block}
    \end{blockarray}
\]

\end{ex}

Now that each tier has its own rate matrix, we can construct transition matrices for each tier.
The Kronecker sum and product are related by the following equation \cite{kronecker_maatrix_book}, mimicking exponentiation properties of numbers:
\[
\exp( A \oplus B)  =  \exp(A) \otimes \exp(B). 
\] 
So, letting $Q_l$ be the rate matrix for each tier $l$, the transition matrix is $P_l(t)=\exp{(Q_l t)}$. Suppose $Q_l=q_1\oplus q_2$ with $\exp(q_1 t)=P_1(t)$ and $\exp(q_2 t)=P_2(t)$. Due to the bilinearity of the Kronecker product and the exponentiation property above,
\begin{align*}
    \exp{(Q_l t)} &= \exp{([q_1 \oplus q_2]t)} \\
    &=\exp{([q_1\otimes I_n]t+[I_m\otimes q_2]t)} \\
    &=\exp{([q_1 t]\otimes I_n + I_m\otimes [q_2 t])} \\
    &= \exp{([q_1 t]\oplus [q_2 t])} \\
    &=\exp{(q_1 t)}\otimes\exp{(q_2 t)}\\
    &= P_1(t)\otimes P_2(t),
\end{align*}

where $I_n$ is an identity matrix the size of $q_2$ and $I_m$ is an identity matrix the size of $q_1$.

Since each tier now has a transition matrix with size dependent of the number of edges in that tier, we cannot multiply the  matrices across tiers directly. 
At any speciation event, we increase the length of the state sequence by one. To accomplish this, we take the state of the vertex with two children and duplicate it.  This makes sense because at that moment there are two identical copies which then evolve separately in the Markov process for the
next tear.  For example, in Figure \ref{fig:levels}, if $Y_1$ is in state $a$ and $Y_2$ is in state $b$, then the sequence at the end of the first tier is $ab$, and the sequence at the start of the second tier is $aab$. Therefore, we only take the rows of the probability matrix corresponding to sequences where the diverging lineages have the same state.

\begin{defn}
Let $T$ be a tiered X-tree. We call two edges of $T$ \textit{sibling edges} if they are in the same tier and share their most recent ancestor with out-degree two. Similarly, two vertices are called \textit{sibling vertices} if they are in the same tier and they share their most recent ancestor with out-degree two. 
We call the set of all pairs of sibling vertices and edges $\sib(T)$. 
\end{defn}

For each pair of sibling edges, we will combine the rate matrices for each edge using the Kronecker sum. After summing the rate matrices for the sibling edges, then exponentiating them, we will truncate the matrix in order to have a transition matrix of the appropriate dimension. Then we can proceed to multiply this transition matrix with those of the other edges in the tier.

\begin{defn}
    Let $T$ be a tiered X-tree with $r$ possible states on the random variables, and let   $a\rightarrow b$ and $a\rightarrow c$
    be two sibling edges. 
    Let $P_{b}$ and $P_{c}$ be their respective transition matrices, and let $P_{bc}=P_b\otimes P_c$. We will index $P_{bc}$ with ordered pairs of states. The \textit{truncation} of $P_{bc}$, denoted $\trunc(P_{bc})$ is the matrix containing of only the rows of $P_{bc}$ indexed by a pair with both states equal. 
    
\end{defn}

\begin{ex}
   In Example \ref{Q1}, we calculated the rate matrix for the first tier of the tree in Figure \ref{fig:levels}. Suppose the possible states are $0$ and $1$, and the transition matrix for the same tier is the following: 
    \[ \exp{(Q_1t)}=
\begin{blockarray}{ccccc}
        & 00 & 01 & 10 & 11 \\
      \begin{block}{c[cccc]}
        00 & a & b & c & d \\
        01 & e & f & g & h \\
        10 & i & j & k & l \\
        11 & m & n & o & p \\
      \end{block}
    \end{blockarray},
\]

with each row labeled by the sequence at the beginning of the tier and each column labeled by the sequence at the end of the tier. To construct $P$, we will take only the first and fourth rows of the matrix, since those rows represent the state at $Y_0$ having been duplicated:
 \[ \trunc(\exp{(Q_1t)})=P=
\begin{blockarray}{ccccc}
        & 00 & 01 & 10 & 11 \\
      \begin{block}{c[cccc]}
        00 & a & b & c & d \\
        11 & m & n & o & p \\
      \end{block}
    \end{blockarray}   .
\]
\end{ex}

Once the truncated transition matrices for the sibling edges are incorporated into the larger tier transition matrix, we multiply the transition matrix for each tier in chronological order to get the probabilities of the states on the leaves. Let $T$ be a tiered $X$-tree, 
with root $Y_0$ and leaf set $\{Y_{m+1},\ldots,Y_n\}$. 
We will take the first two internal vertices at the conclusion of the first tier, $Y_1$ and $Y_2$, and combine them into a random vector $X_1 = (Y_1, Y_2)$, 
describing the  states across the nodes. 
We repeat this with each tier, so for a tree with $l$ tiers, the interior variables become $X_1, \ldots,X_{l-1}$ with the leaves grouped in the variable $X_l$. Then the probability on the leaves is

\[
    \pp(X_l=x_l)= \sum_{x_0\in S}\cdots \sum_{x_{l-1} \in S^{l}} \pp(X_0=x_0)\prod_{1\leq i \leq l} \pp(X_i=x_i |X_{i-1} = x_{i-1})   
\]
where the lowercase $x_0, \ldots ,x_l$ represent sequences of states of their respective variables.

\begin{ex}
    \label{ex: one distribution}
    We return to the tree on four leaves in Figure \ref{fig:kronecker}.
    Assume that all of the random variables $Y_i$ are binary. 
    In this tree, 
    we have $X_1  = (Y_1, Y_2)$, 
     $X_2 = (Y_3, Y_4, Y_5)$, and  $X_5 = (Y_6, Y_7, Y_8, Y_9)$. 
     Let $\pi$ be the $1 \times 2$ vector containing the probability of states at $X_0 = Y_0$. 
     Then the probability distribution of states on the leaves is 
    \[    \pi  (\trunc(P_1\otimes P_2))(\trunc(P_3 \otimes P_4) \otimes P_5 )(P_6 \otimes \trunc(P_7 \otimes P_8 )\otimes P_9).    \]
    This product is a $1 \times 2$ matrix, multiplied by a $2\times 4$ matrix, a $4 \times 8$ matrix, and an $8 \times 16$ matrix, so the product is well defined and is a vector of length 16. Using these variables to represent sequences of states, the probability of the leaves having the sequence $x_3$ is 
    
    \begin{align*}
    \pp(X_3=x_3)= \sum_{x_0\in S}\sum_{x_1\in S^2} \sum_{x_2 \in S^3} \pp(X_0=x_0) \pp(X_1=x_1|X_0 = x_0)\pp(X_2=x_2|X_1 = x_1) \pp(X_3=x_3|X_2=x_2).     
\end{align*}
\end{ex}

\begin{figure}[]
\centering
\resizebox{.5\textwidth}{!}{%
\begin{circuitikz}
\tikzstyle{every node}=[font=\huge]

\draw [ fill={rgb,255:red,0; green,0; blue,0} ] (-2.5,13.25) circle (0.25cm) node [left=7pt]{$0$};

\draw [line width=2pt, short] (-2.5,13.25) -- (-5,9.5) node [midway, fill=white] {$P_1$};
\draw [line width=2pt, short] (-2.5,13.25) -- (0,9.5) node [midway, fill=white] {$P_2$};

\draw [ fill={rgb,255:red,0; green,0; blue,0} ] (-5,9.5) circle (0.25cm) node [left=7pt]{$1$};
\draw [ fill={rgb,255:red,0; green,0; blue,0} , line width=0.2pt ] (-2.5,5.75) circle (0.25cm) node [left=7pt]{$4$};

\draw [line width=2pt, short] (-5,9.5) -- (-7.5,5.75) node [midway, fill=white] {$P_3$};
\draw [line width=2pt, short] (-5,9.5) -- (-2.5,5.75) node [midway, fill=white] {$P_4$};

\draw [line width=2pt, short] (-2.5,5.75) -- (-5,2) node [midway, fill=white] {$P_7$};
\draw [line width=2pt, short] (-2.5,5.75) -- (0,2) node [midway, fill=white] {$P_8$};

\draw [line width=2pt, short] (0,9.5) -- (2.5,5.75) node [midway, fill=white] {$P_5$};
\draw [line width=2pt, short] (-7.5,5.75) -- (-10,2) node [midway, fill=white] {$P_6$};
\draw [line width=2pt, short] (2.5,5.75) -- (5,2) node [midway, fill=white] {$P_9$};

\draw [ fill={rgb,255:red,0; green,0; blue,0} , line width=0.2pt ] (0,9.5) circle (0.25cm) node [left=7pt]{$2$};
\draw [ fill={rgb,255:red,0; green,0; blue,0} , line width=0.2pt ] (-7.5,5.75) circle (0.25cm) node [left=7pt]{$3$};
\draw [ fill={rgb,255:red,0; green,0; blue,0} , line width=0.2pt ] (2.5,5.75) circle (0.25cm) node [left=7pt]{$5$};

\draw [ fill={rgb,255:red,0; green,0; blue,0} , line width=0.2pt ] (-10,2) circle (0.25cm) node [left=7pt]{$6$};
\draw [ fill={rgb,255:red,0; green,0; blue,0} , line width=0.2pt ] (0,2) circle (0.25cm) node [left=7pt]{$8$};
\draw [ fill={rgb,255:red,0; green,0; blue,0} , line width=0.2pt ] (-5,2) circle (0.25cm) node [left=7pt]{$7$};
\draw [ fill={rgb,255:red,0; green,0; blue,0} , line width=0.2pt ] (5,2) circle (0.25cm) node [left=7pt]{$9$};

\draw [dashed] (-14,9.5) -- (5,9.5);
\draw [dashed] (-14,5.75) -- (5,5.75);
\node [font=\LARGE] at (-12,11.75) {tier 1};
\node [font=\LARGE] at (-12,7.75) {tier 2};
\node [font=\LARGE] at (-12,3.75) {tier 3};

\end{circuitikz}
}%
\caption{Each edge of the tree has an associated transition matrix, $P_i$. The dashed lines mark the boundary between tiers.}
\label{fig:kronecker}
\end{figure}
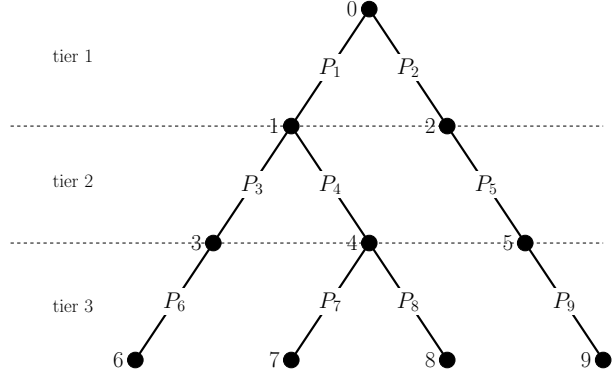


\section{The Local Gene Transfer Model}  \label{sec:local}

In certain biological contexts it is possible for genetic information to be transferred ``laterally" between lineages. 
While mutation occurs within a single lineage, lateral gene transfer, or LGT, describes when a living organism  transfers genetic data directly to another organism \cite{Sieber2017-vl}.
Several authors have described models of random LGT in which the number  of LGT events is Poisson distributed with some rate $\lambda$ \cite{Likelihood_framework_hgt}. 
In these models we assume that when an LGT event occurs, 
one lineage sends a copy of a gene to replace a gene in a different lineage, while leaving the copy of the gene in the original lineage intact.
Furthermore, only one LGT event is allowed to occur at one time. We can visualize this process on the graph of a tree $T$ by drawing a directed edge $\sigma$ between two contemporaneous points $v$ and $v'$. 
We will create a sequence of these LGT events, starting with the most recent and ending with the event closest to the root: 
$\sigma_0,\sigma_1,\ldots,\sigma_k$ where $\sigma_i$ is the edge $(v_i\rightarrow v_i')$. Now by reversing along the direction of the edges from each leaf to the root, we can construct a gene tree. The gene tree displays the evolution of the gene as opposed to the evolution of the species. For each $\sigma_i$ delete the section of the edge immediately above $v_i'$ and take the minimal connected subgraph to be the gene tree. In this way, given a species tree, we can generate random gene trees, and use them as a data pool to investigate identifiability of the species tree.

In this paper, since we are observing the evolution of gene sequences on trees, we will use a slightly different viewpoint of generating gene trees. We will keep the same assumptions that no two transfer events occur at the same time, and that one lineage will ``donate" a copy of its own gene to another lineage. Transfer events will occur at a rate $\lambda$, but $\lambda$ does not need to be constant across
the tree. We will use $\lambda(a,b)$, which is a function of $a$ and $b$, the origin species and destination species. 
We will call this model the Lateral Gene Transfer Model. 

\begin{ex}
    Figure \ref{fig:lgt close up} shows a picture of five edges of a tree in the same tier. There is lateral gene transfer between lineages 1 and 2, represented by the arrow labeled $\lambda(1,2)$. This means that the state $a$ from lineage $1$ was transferred onto lineage $2$, overwriting the previous state $c$. 
\end{ex}

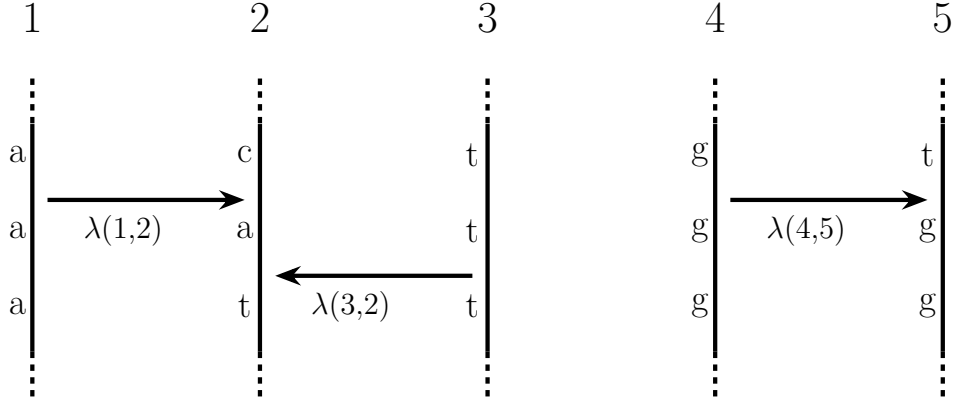
\begin{figure}[]
\centering
\resizebox{.8\textwidth}{!}{%
\begin{circuitikz}
\tikzstyle{every node}=[font=\huge]

\draw [ line width=2pt](-8.75,9.5) to[short] (-8.75,5.75);
\draw [ line width=2pt](-5,9.5) to[short] (-5,5.75);
\draw [ line width=2pt](-1.25,9.5) to[short] (-1.25,5.75);
\draw [ line width=2pt](2.5,9.5) to[short] (2.5,5.75);
\draw [line width=2pt, dashed] (-8.75,10.25) -- (-8.75,9.5);
\draw [line width=2pt, dashed] (-5,10.25) -- (-5,9.5);
\draw [line width=2pt, dashed] (-1.25,10.25) -- (-1.25,9.5);
\draw [line width=2pt, dashed] (2.5,10.25) -- (2.5,9.5);
\draw [line width=2pt, dashed] (-8.75,5.75) -- (-8.75,5);
\draw [line width=2pt, dashed] (-5,5.75) -- (-5,5);
\draw [line width=2pt, dashed] (-1.25,5.75) -- (-1.25,5);
\draw [line width=2pt, dashed] (2.5,5.75) -- (2.5,5);
\node [font=\huge] at (-8.75,11.25) {1};
\node [font=\huge] at (-5,11.25) {2};
\node [font=\huge] at (-1.25,11.25) {3};
\node [font=\huge] at (2.5,11.25) {4};
\node [font=\LARGE] at (-9,9) {a};
\node [font=\LARGE] at (-5.25,9) {c};
\node [font=\LARGE] at (-1.5,9) {t};
\node [font=\LARGE] at (2.25,9) {g};
\draw [line width=2pt, ->, >=Stealth] (-8.5,8.25) -- (-5.25,8.25);
\node [font=\LARGE] at (-5.25,7.75) {a};
\draw [line width=2pt, ->, >=Stealth] (2.75,8.25) -- (6,8.25);
\draw [line width=2pt, ->, >=Stealth] (-1.5,7) -- (-4.75,7);
\node [font=\LARGE] at (-1.5,7.75) {t};
\node [font=\LARGE] at (-5.25,6.5) {t};
\node [font=\LARGE] at (-9,7.75) {a};
\node [font=\LARGE] at (-9,6.5) {a};
\node [font=\LARGE] at (-1.5,6.5) {t};
\node [font=\LARGE] at (2.25,7.75) {g};
\node [font=\LARGE] at (2.25,6.5) {g};
\node [font=\Large] at (-3.5,6.5) {$\lambda$(3,2)};
\node [font=\Large] at (-7.25,7.75) {$\lambda$(1,2)};
\node [font=\Large] at (4,7.75) {$\lambda$(4,5)};
\draw [ line width=2pt](6.25,9.5) to[short] (6.25,5.75);
\draw [line width=2pt, dashed] (6.25,10.25) -- (6.25,9.5);
\draw [line width=2pt, dashed] (6.25,5.75) -- (6.25,5);
\node [font=\huge] at (6.25,11.25) {5};
\node [font=\LARGE] at (6,9) {t};
\node [font=\LARGE] at (6,7.75) {g};
\node [font=\LARGE] at (6,6.5) {g};
\end{circuitikz}
}%

\caption{In this illustration there is lateral gene transfer 
between lineages 1 and 2, such that the state \textbf{a} on lineage 1 overwrites the state \textbf{c} on lineage 2. Similarly, there is LGT between lineages 4 and 3, as well as between 4 and 5.}
\label{fig:lgt close up}
\end{figure}

 Under the Lateral Gene Transfer Model we will specify the possible LGT events with a graph describing the relationships between edges. 

\begin{defn}
    Let $G_i=(V,E)$ be a directed graph with vertex set $V$ being the set of species at tier $i$. 
    We will call this graph the \textit{gene transfer graph} associated to a tree $T$ on tier $i$. The edge set $E$ consists of ordered pairs of species $(u,v)$ such that the state of $u$ can overwrite the state of $v$ in a lateral gene transfer event, for any pair of states. 
\end{defn}

    For each edge of the graph $G_i$, let  $\lambda(u,v)\in \rr_{>0}$ be the parameter which describes the rate at which 
    lineage $u$ overwrites the state on lineage $v$ with the state on $u$. 
    Let $\Lambda_i$ be the set of all $\lambda$ parameters for the graph $G_i$.

\begin{ex}
    Figure \ref{fig:gene transfer graph} shows a possible gene transfer graph for the tier shown in Figure \ref{fig:lgt close up}.  Lineage $1$ is able to overwrite the states on either $2$ or $3$, but $2$ and $3$ can only overwrite the states on each other. 
\end{ex}

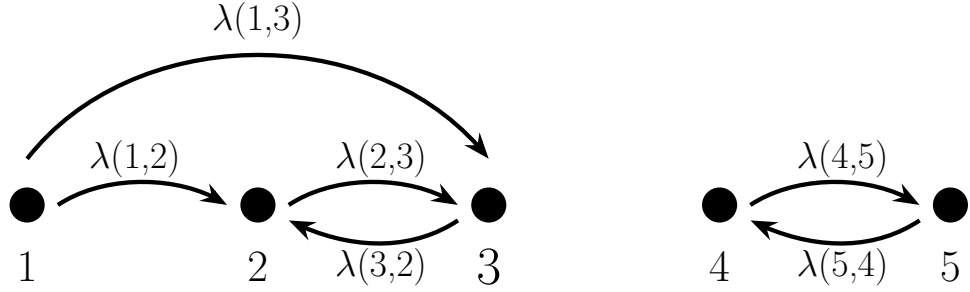
\begin{figure}[]
\centering
\resizebox{.8\textwidth}{!}{%
\begin{circuitikz}
\tikzstyle{every node}=[font=\LARGE]

\node [font=\huge] at (-8.75,2.25) {1
};
\node [font=\huge] at (-5,2.25) {2
};
\node [font=\Huge] at (-1.25,2.25) {3
};
\node [font=\huge] at (2.5,2.25) {4
};
\draw [line width=2pt, ->, >=Stealth] (-8.25,3.25) .. controls (-7.5,3.75) and (-6.25,3.75) .. (-5.5,3.25) ;
\node [font=\LARGE] at (-3,4) {$\lambda$(2,3)};
\node [font=\LARGE] at (-7,4) {$\lambda$(1,2)};
\node [font=\LARGE] at (4.5,4) {$\lambda$(4,5)};
\node [font=\huge] at (6.25,2.25) {5};
\draw [ fill={rgb,255:red,0; green,0; blue,0} , line width=2pt ] (-8.75,3.25) circle (0.25cm);
\node [font=\huge] at (-8.5,3) {};
\node [font=\huge] at (-8.5,3) {};
\draw [ fill={rgb,255:red,0; green,0; blue,0} , line width=2pt ] (-5,3.25) circle (0.25cm);
\draw [ fill={rgb,255:red,0; green,0; blue,0} , line width=2pt ] (-1.25,3.25) circle (0.25cm);
\draw [ fill={rgb,255:red,0; green,0; blue,0} , line width=2pt ] (2.5,3.25) circle (0.25cm);
\draw [ fill={rgb,255:red,0; green,0; blue,0} , line width=2pt ] (6.25,3.25) circle (0.25cm);
\draw [line width=2pt, ->, >=Stealth] (-4.5,3.25) .. controls (-3.75,3.75) and (-2.5,3.75) .. (-1.75,3.25) ;
\draw [line width=2pt, ->, >=Stealth] (3,3.25) .. controls (3.75,3.75) and (5,3.75) .. (5.75,3.25) ;
\draw [line width=2pt, ->, >=Stealth] (-1.75,3) .. controls (-2.5,2.5) and (-3.5,2.5) .. (-4.5,3) ;
\node [font=\LARGE] at (-3,2.25) {$\lambda$(3,2)};
\draw [line width=2pt, ->, >=Stealth] (5.75,3) .. controls (5,2.5) and (4,2.5) .. (3,3) ;
\node [font=\LARGE] at (4.5,2.25) {$\lambda$(5,4)};
\draw [line width=2pt, ->, >=Stealth] (-8.75,4) .. controls (-7,6.25) and (-3,6.25) .. (-1.25,4) ;
\node [font=\LARGE] at (-5,6.25) {$\lambda$(1,3)};
\end{circuitikz}
}%

\caption{A possible graph describing the allowed lateral gene transfer events between the lineages of Figure 2. Note that the edges need not be next to each other for gene transfer to be allowed. Gene transfer can also be allowed in one direction, but not the other.}
\label{fig:gene transfer graph}
\end{figure}

When constructing the rate matrix for all edges at a single tier $l$, each connected component of $G_l$ must be constructed individually.

\begin{defn}
\label{def: construct rate matrix}
    Let $T$ be a tiered X-tree and let $G_l$ be the gene transfer graph for a tier $l$. Assume that there are $n$ lineages.  Let $i = (i_1, \ldots, i_n) \in S^n$
    and $j = (j_1, \ldots, j_n) \in S^n$.  
    Then the rate matrix associated to the $l$-th tier is constructed as follows:
\[     
Q_l(i;j) = 
\begin{cases}
    q^{(a)}_{i_a,j_a} + \sum_{a' :  i_{a'} = j_a,  (a', a) \in G_l}  \lambda(a',a)  &\mathrm{d}_H(i,j)=1,  i_a \neq j_a \\
    -\sum_{k \in S^n:k\neq i}{Q(i;k)} & \mathrm{d}_H(i,j)=0\\
    0 &\mathrm{d}_H(i,j)>1
\end{cases}   
\]
\end{defn}

In words, for off diagonal entries, the rate $Q_l(i;j)$ is nonzero when $i$ and $j$
differ in exactly one position.  The resulting rate is the sum of rates for 
all ways that a change could occur changing $i$ into $j$.  If that change position is the $a$th position,  this consists
of the site substitution rate $q^{(a)}_{i_a, j_a}$
plus the sum of $\lambda(a', a)$ rates for each possible LGT event that
could change $i$ into $j$.  

\begin{ex}
     We return to Example \ref{ex: one distribution} and the tree on four leaves in Figure \ref{fig:kronecker}.  Instead of separate transition matrices for sibling edges, we will use a single matrix which includes lateral gene transfer. Edges 1 and 2 will have transition matrix $P_{12}$, edges 3 and 4 will have transition matrix $P_{34}$, and edges 7 and 8 will have transition matrix $P_{78}$. Then the probability distribution of states on the leaves is 
    \[    \pi (\trunc(P_{12}))(\trunc(P_{34}) \otimes P_5 )(P_6 \otimes \trunc(P_{78} )\otimes P_9).    \]
\end{ex}

To define the model, we need to condense the notation of parameters. 
Since we consider equidistant trees, we can consider the length parameters
on the tree's edges to just consist of one parameter $t_i$ for each tier. 
We will denote the vector of tier lengths as 
$t=(t_1,\ldots,t_k)$, $t\in \rr^k$. 
Additionally, let $T$ have an associated sequence of gene transfer 
graphs $G=(G_{1},...,G_{k})$. 
Let $\Lambda=(\Lambda_{1},...,\Lambda_{k}), \Lambda_l \in \rr^{n_l}_{>0}$ be the set of edge parameters for each gene transfer graph, 
where $n_l$ is the number of edges in $G_l$. Let $E(G)$ be the total number of edges in $G$.
We associate to each edge $e$ in the tree a rate matrix $Q^{(e)}$. 
Let the sequence of transition matrices on each edge be $Q$. 

\begin{defn}
    Let there be a set of taxa $X$ and let $T$ be a tiered phylogenetic $X$-tree. The \textit{Lateral Gene Transfer Model} on $T$ is the set of probability distributions on the sequences of states on the taxa $x_1, \ldots  ,x_n$ resulting from the Markov process: 
    \[
    M(T,G)=\{\pp(T,G,\Lambda,Q,t):\Lambda\in\rr^{E(G)}_{>0},Q\in \rr^{k\times(mr\times mr)},t\in \rr^k_{>0}\}.
    \]
\end{defn}

The model of lateral gene transfer used in \cite{kubatko2026revisitingrandommodellateral} is in fact a special case of the Lateral Gene Transfer Model described here. This case uses a uniform rate matrix across all edges, a uniform LGT rate $\lambda$, and a gene transfer graph in which the subgraph for each tier is a complete bidirectional graph. 

Allowing lateral gene transfer to occur between any edges in the same tier is possibly an unrealistic assumption. We might expect gene transfer events only to happen between closely related species \cite{limited_range_of_lgt}. Because of this assumption, it follows that we should create a model in which the possible LGT events are restricted in some way.

\begin{defn}
    Let $T$ be a tiered X-tree. For each tier $l$, let $G_l$ be the graph consisting of all edges $(a_1\rightarrow a_2)$ and $(a_2 \rightarrow a_1)$ for all sibling edges 
    $a_1$ and $a_2$ in tier $l$. 
    Let $G^{loc} = (G_1, \ldots, G_k)$ be the resulting gene transfer graph. 
    We call $G^{loc}$ the \textit{local gene transfer graph}. The \textit{Local Gene Transfer Model} is the set of probability distributions  $M(T,G^{loc})$, where $\Lambda\in \rr^c_{>0}$ where $c$ is the number of connected components of $G^{loc}$. 
\end{defn}

Figure \ref{fig: cherry structure} shows an example of where gene transfer is allowed to occur within the tree under the new restrictions. Figure \ref{fig: local gene transfer graph} shows the corresponding local gene transfer graph. Under this model, only the most closely related lineages can exchange genetic material.

Identifiability analysis is the problem of determining which parameters of the model
can be determined from data.  In that case of our lateral gene transfer model, we would be
interested in determining the tree, rate parameters, and branch lengths of the model.
We assume that data consists of i.i.d.~samples from a probability distribution $p \in
M(T,G)$.  As the sample size tends to infinity, the empirical distribution obtained from this
sample converges to $p$.  Hence,  structural identifiability questions concern whether
or not the probability distribution $p$ has enough information to infer the parameters
of the model.  In this work we specifically focus on the tree parameter $T$ itself,
which leads us the the following definition.

\begin{defn}
    Let $T_1$ and $T_2$ be binary tiered $X$-trees with differing unrooted topologies.  
    Let $T'_1$ and $T'_2$ be the unrooted $X$-trees obtained by suppressing the root
    and all vertices of degree $2$ in $T_1$ and $T_2$ respectively.  
    The unrooted topologies $T'_1$ and $T'_2$ are 
     \textit{generically distinguishable} if 
\[  \dim[M(T_1,G_1^{loc})\cap M(T_2,G_2^{loc})]< \min \left( \dim[M(T_1,G_1^{loc})], \dim[M(T_2,G_2^{loc})]\right) .\]
The unrooted tree parameter is \emph{generically identifiable}  for all pairs $T_1$ and $T_2$ that have different
unrooted $X$-trees.  
\end{defn}

The modifier ``generically'' is added to account for the fact that the models
$M(T_1,G_1^{loc})$ and $M(T_2,G_2^{loc})$ might have a nontrivial overlap.  However, if the
dimension of the intersection is small, then, with probability $1$ any distribution that is
in one of the models does not belong to both. 

\begin{figure}[]
\centering
\subfloat[]{
\resizebox{.5\textwidth}{!}{%
\begin{circuitikz}
\tikzstyle{every node}=[font=\LARGE]

\draw [ fill={rgb,255:red,0; green,0; blue,0} ] (-2.5,13.25) circle (0.25cm);
\draw [ fill={rgb,255:red,0; green,0; blue,0} , line width=0.2pt ] (-5,9.5) circle (0.25cm);
\draw [line width=2pt, short] (-2.5,13.25) -- (-5,9.5);
\draw [line width=2pt, short] (-2.5,13.25) -- (0,9.5);

\draw [ fill={rgb,255:red,0; green,0; blue,0} ] (-5,9.5) circle (0.25cm);

\draw [line width=2pt, short] (-5,9.5) -- (-7.5,5.75);
\draw [line width=2pt, short] (-5,9.5) -- (-2.5,2);

\draw [ fill={rgb,255:red,0; green,0; blue,0} ] (-2.5,2) circle (0.25cm);
\draw [line width=2pt, short] (2.5,5.75) -- (0,2);

\draw [line width=2pt, short] (0,9.5) -- (2.5,5.75);
\draw [line width=2pt, short] (2.5,5.75) -- (5,2);
\draw [line width=2pt, short] (-7.5,5.75) -- (-10,2);
\draw [ fill={rgb,255:red,0; green,0; blue,0} , line width=0.2pt ] (-10,2) circle (0.25cm);
\draw [ fill={rgb,255:red,0; green,0; blue,0} , line width=0.2pt ] (-7.5,5.75) circle (0.25cm);
\draw [ fill={rgb,255:red,0; green,0; blue,0} , line width=0.2pt ] (-3.75,5.75) circle (0.25cm);
\draw [ fill={rgb,255:red,0; green,0; blue,0} , line width=0.2pt ] (2.5,5.75) circle (0.25cm);
\draw [ fill={rgb,255:red,0; green,0; blue,0} , line width=0.2pt ] (0,9.5) circle (0.25cm);
\draw [ fill={rgb,255:red,0; green,0; blue,0} , line width=0.2pt ] (0,2) circle (0.25cm);
\draw [ fill={rgb,255:red,0; green,0; blue,0} , line width=0.2pt ] (5,2) circle (0.25cm);

\draw[line width=2pt, -{Stealth} ] (-3.5,11.5) -- (-1.5,11.5);
\draw[line width=2pt, {Stealth}- ] (-4.2,10.5) -- (-.8,10.5);

\draw[line width=2pt, -{Stealth} ] (-6,7.75) -- (-4.5,7.75);
\draw[line width=2pt, {Stealth}- ] (-6.7,6.75) -- (-4.2,6.75);

\draw[line width=2pt, -{Stealth} ] (-8.5,4) -- (-3.3,4);
\draw[line width=2pt, {Stealth}- ] (-9.2,3) -- (-3,3);

\draw[line width=2pt, -{Stealth} ] (1.5,4) -- (3.5,4);
\draw[line width=2pt, {Stealth}- ] (0.8,3) -- (4.2,3);

\node [font=\Huge] at (-4.2,11.75) {1};
\node [font=\Huge] at (-.8,11.75) {2};

\node [font=\Huge] at (-6.8,7.75) {3};
\node [font=\Huge] at (-3.8,7.75) {4};
\node [font=\Huge] at (2,7.75) {5};

\node [font=\Huge] at (-9.5,3.75) {6};
\node [font=\Huge] at (-2.5,3.75) {7};
\node [font=\Huge] at (.5,3.75) {8};
\node [font=\Huge] at (4.5,3.75) {9};

\end{circuitikz}
\label{fig: cherry structure}
}%
}

\subfloat[]{
\resizebox{.5\textwidth}{!}{%
\begin{circuitikz}
\tikzstyle{every node}=[font=\LARGE]

\node [font=\huge] at (-5,5.25) {1};
\node [font=\huge] at (-1,5.25) {2};

\draw [ fill={rgb,255:red,0; green,0; blue,0} , line width=2pt ] (-5,6) circle (0.15cm); 
\draw [ fill={rgb,255:red,0; green,0; blue,0} , line width=2pt ] (-1,6) circle (0.15cm); 

\draw [line width=2pt, ->, >=Stealth] (-4.5,6.25) .. controls (-3.5,6.75) and (-2.5,6.75) .. (-1.5,6.25) ; 
\draw [line width=2pt, ->, >=Stealth] (-1.5,5.75) .. controls (-2.5,5.25) and (-3.5,5.25) .. (-4.5,5.75) ; 

\node [font=\huge] at (-7,2.25) {3};
\node [font=\huge] at (-3,2.25) {4};
\node [font=\huge] at (1,2.25) {5};

\draw [ fill={rgb,255:red,0; green,0; blue,0} , line width=2pt ] (-7,3) circle (0.15cm); 
\draw [ fill={rgb,255:red,0; green,0; blue,0} , line width=2pt ] (-3,3) circle (0.15cm); 
\draw [ fill={rgb,255:red,0; green,0; blue,0} , line width=2pt ] (1,3) circle (0.15cm); 

\draw [line width=2pt, ->, >=Stealth] (-6.5,3.25) .. controls (-5.5,3.75) and (-4.5,3.75) .. (-3.5,3.25) ; 
\draw [line width=2pt, ->, >=Stealth] (-3.5,2.75) .. controls (-4.5,2.25) and (-5.5,2.25) .. (-6.5,2.75) ; 

\node [font=\huge] at (-9,-0.75) {6};
\node [font=\huge] at (-5,-0.75) {7};
\node [font=\huge] at (-1,-0.75) {8};
\node [font=\huge] at (3,-0.75) {9};

\draw [ fill={rgb,255:red,0; green,0; blue,0} , line width=2pt ] (-9,0) circle (0.15cm); 
\draw [ fill={rgb,255:red,0; green,0; blue,0} , line width=2pt ] (-5,0) circle (0.15cm); 
\draw [ fill={rgb,255:red,0; green,0; blue,0} , line width=2pt ] (-1,0) circle (0.15cm); 
\draw [ fill={rgb,255:red,0; green,0; blue,0} , line width=2pt ] (3,0) circle (0.15cm); 

\draw [line width=2pt, ->, >=Stealth] (-8.5,0.25) .. controls (-7.5,0.75) and (-6.5,0.75) .. (-5.5,0.25) ; 

\draw [line width=2pt, ->, >=Stealth] (-5.5,-0.25) .. controls (-6.5,-0.75) and (-7.5,-0.75) .. (-8.5,-0.25) ; 

\draw [line width=2pt, ->, >=Stealth] (-0.5,0.25) .. controls (0.5,0.75) and (1.5,0.75) .. (2.5,0.25) ; 

\draw [line width=2pt, ->, >=Stealth] (2.5,-0.25) .. controls (1.5,-0.75) and (0.5,-0.75) .. (-0.5,-0.25) ; 

\end{circuitikz}
}%
\label{fig: local gene transfer graph}
}

\caption{Figure (A) shows an example tiered X-tree with labeled edges. The arrows represent gene transfer events and can only take place between two lineages which share their most ancestor with out-degree two. Figure (B) shows the local gene transfer graph describing where lateral gene transfer is possible on the tree shown in Figure (A), under the Local Gene Transfer Model.}

\end{figure}
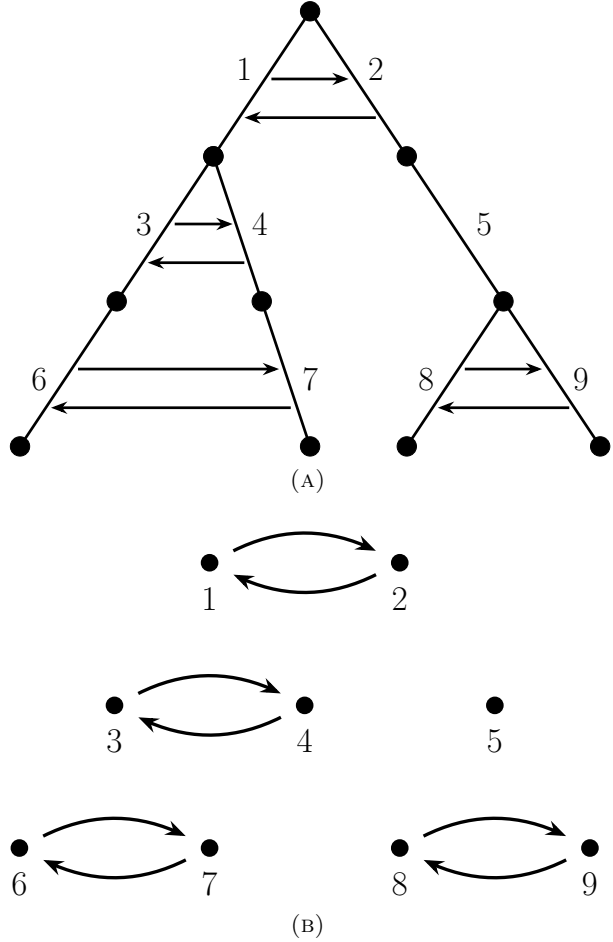


\section{Reinterpretation of the local model as a graphical model}\label{sec:graphical}

Graphical models are statistical models that use graphs to represent
interactions between collections of random variables.
Given a graph $G = (V,E)$, the vertices represent  random variables, 
and the edges represent dependence relations between the random variables. 
In a typical phylogenetic tree model with no lateral gene transfer, the vertices of the tree are the random variables of the graphical model which depend only on their parent nodes. 
This is an example of a directed graphical model. 
Other models of phylogenetic evolution cannot be viewed as graphical models themselves, but can be modified into an associated graphical model. 
For example, when hybridization or lateral gene transfer events occur, they introduce dependencies that are not strictly ancestor-to-descendant. This complicates the relationships in the model. Still, an associated graphical model can be useful for extracting relationships between the probabilities of the ultimate taxa.
An example of this idea for the displayed tree phylogenetic network model 
appears in \cite{Sullivant2025}.  

We start by defining key pieces to build the graphical models.  
Graphical models can be defined in two ways, either through
conditional independence relations or parametrically.

Let $G = (V,E)$ be a  directed acyclic graph.
We let $X=(X_v  \mid  v\in V)$ be a random vector, whose components
are indexed by the vertices of $V$.  For a subset $A \subseteq V$, 
let $X_A = (X_a | a \in A)$ denote the subvector of $X$ indexed by $A$. For two subsets of $V$, $A$ and $B$, the joint probability density function for the random vectors $X_A$ and $X_B$ is $f_{A\cup B}(x_A,x_B)$. Adding a $C$, third subset of $V$, $f_{A\cup B |C}(x_A,x_B|x_C)$ is the joint probability distribution on variables $X_A$ and $X_B$ conditioned on the variable $X_C$.

\begin{defn}
    Let $A,B,C \subseteq V$ be pairwise disjoint. 
    The random vector $X_A$ is \textit{conditionally independent of $X_B$ given $X_C$} if and only if
    \[
    f_{A\cup B| C}(x_A,x_B|x_C)=f_{A|C}(x_A|x_C)\cdot f_{B|C}(x_B|x_C)   
    \]
    for all $x_A,x_B,$ and $x_C$, where $f$ is the joint probability density function.
    We use the notation \\ $X_A \indep X_B|X_C$ to denote that the random vector $X$ satisfies the conditional independence statement that $X_A$ is conditionally independent from $X_B$ given $X_C$.
\end{defn}

The graphical model associated with $G$ encodes the dependence relations between two variables with an edge of the graph incident to or a path between the two nodes. If there are no edges between the nodes, conditional independence is implied between them. The \textit{directed local Markov property} associated to a graph describes the distributions such that each variable depends only on its parent.

\begin{defn}
    Let $G=(V,E)$ be a directed acyclic graph. 
    The \textit{directed local Markov property} associated with $G$ consists of all conditional independence statements $X_v \indep X_{\nd(v)\setminus \pa(v)} | X_{\pa(v)}$, for all $v\in V$.
\end{defn}

In order to make use of the directed local Markov property it is helpful to discuss the parametric description of graphical models, which in this application is adept in describing the underlying processes by which the data is generated. In the case of directed graphs, each vertex should depend only on its parent vertex, and the probability densities should factorize in the same way. Let $T=(V,E)$ be a directed acyclic graph. For each node $j\in V$ we have a conditional distribution variable $X_j$ conditional on the parents of node $j$ in the graph $T$, $f_j (X_j|x_{\pa(j)})$, and consider the probability densities of the form 
\[f(x) = \prod_{j\in V}f_j(x_j|x_{\pa(j)}).\]
This is called the recursive factorization property.

\begin{defn}
    The \textit{parametric directed graphical model} associated to the directed acyclic graph $T$ consists of all probability density functions that factorize as the product of conditionals:
    \[f(x) = \prod_{j\in V}f_j(x_j|x_{\pa(j)})\]
\end{defn}

It turns out that satisfying the local directed Markov property is equivalent to satisfying the recursive factorization property.

\begin{thm}\cite[Theorem 13.2.10]{Algebraic_Statistics}\label{local_markov_factorization}
    A probability density function $f$ satisfies the recursive factorization property associated to the directed acyclic graph $T$ if and only if it satisfies the directed local Markov property associated to $T$. 
    
\end{thm}

The probability densities in the local lateral gene transfer model model 
do not satisfy the parametric factorization property, 
since edges may affect the states of their sibling lineages. 
Ideally, we would like to have a graph in which the probabilities do satisfy the recursive factorization property. 
Then we can use the conditional independence relations from flattening matrices to separate the leaves into splits. 
Once we have all the splits of the tree, we will be able to reconstruct the unrooted tree structure from the data on the leaves, via the Splits Equivalence Theorem (Theorem \ref{splits_theorem}).

For a given tiered $X$-tree $T$,
we will construct a new tree $\pinch(T)$, on which we can construct
a graphical model that has the same distribution. 
The basic idea is as follows:
suppose there is a sequence of sibling edges, with each pair of sibling edges descending from a pair of edges in the previous tier. Let the first pair of sibling edges on vertices $a,b,c,$ be $a\rightarrow b$ and $a\rightarrow c$, let an intermediary pair of edges be $s\rightarrow u$ and $t\rightarrow v$, and let the final pair be $w\rightarrow y$ and $x\rightarrow z$. We will collapse each pair of sibling edges together. 
For the first pair, we will introduce a new vertex $bc$. Then we collapse the two edges on $a,b,$ and $c$ to one edge $a\rightarrow bc$. For the intermediary edges, we replace the pair $s\rightarrow u$ and $t\rightarrow v$ with the edge $st\rightarrow uv$ on vertices $st$ and $uv$. The final pair we will call \textit{terminal vertices}, since they are sibling vertices whose children are not sibling vertices. The set of pairs of terminal vertices and pairs of terminal edges in a tree $T$ is called $\term(T)$. For this pair of terminal vertices, we similarly replace edges $w\rightarrow y$ and $x\rightarrow z$ with an edge $wx\rightarrow yz$, with $wx$ and $yz$. However, we will also add edges $yz\rightarrow y'$ and $yz\rightarrow z'$, where $y'$ and $z'$ each take the states of $y$ and $z$ in $T$, and are thus completely determined by vertex $yz$. The rest of the edges and vertices of T remain unchanged. We will call this the \textit{pinched graph of  $T$}, or $\pinch(T)$.

The definition of the pinched graph is formalized in the following:

\begin{defn}
    For a rooted phylogenetic X-tree $T=(V,E)$, the graph $\pinch(T)=(V',E')$ is the graph where
    \begin{itemize}
        \item $V'= \left( V\cup   \{ ab: (a,b)\in \sib(T) \} \right) \setminus \{ c,d: (c,d)\notin \term(T), (c,d)\in \sib(T)  \}    $
        \item $E'=  \left( E \cup \{ ab \rightarrow cd: (a\rightarrow c , b\rightarrow d) \in \sib(T) ,a \rightarrow bc: (a\rightarrow b, a\rightarrow c) \in \sib(T) \}  \right) \\ \setminus \{ u\rightarrow w,v\rightarrow x: (u\rightarrow w , v\rightarrow x )\notin \term(T),(u\rightarrow w , v\rightarrow x )\in \sib(T)  \}  $
    \end{itemize}
 \end{defn}

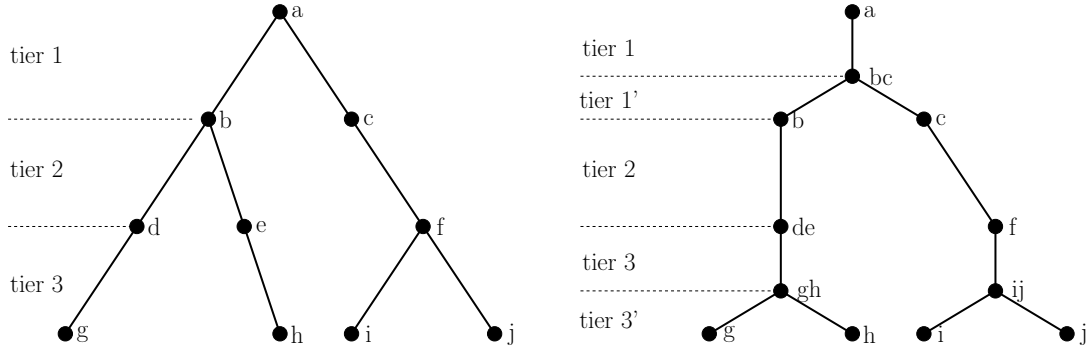
\begin{figure}[]
\centering
\resizebox{.9\textwidth}{!}{%
\begin{circuitikz}
\tikzstyle{every node}=[font=\LARGE]

\draw [ fill={rgb,255:red,0; green,0; blue,0} ] (-2.5,13.25) circle (0.25cm) ;
\node [font=\Huge] at (-1.9,13.25) {a};
\draw [ fill={rgb,255:red,0; green,0; blue,0} , line width=0.2pt ] (-5,9.5) circle (0.25cm);
\node [font=\Huge] at (-4.4,9.5) {b};
\draw [line width=2pt, short] (-2.5,13.25) -- (-5,9.5);
\draw [line width=2pt, short] (-2.5,13.25) -- (0,9.5);

\draw [line width=2pt, short] (-5,9.5) -- (-7.5,5.75);
\draw [line width=2pt, short] (-5,9.5) -- (-2.5,2);

\draw [ fill={rgb,255:red,0; green,0; blue,0} ] (-2.5,2) circle (0.25cm);
\node [font=\Huge] at (-1.9,2) {h};
\draw [line width=2pt, short] (2.5,5.75) -- (0,2);

\draw [line width=2pt, short] (0,9.5) -- (2.5,5.75);
\draw [line width=2pt, short] (2.5,5.75) -- (5,2);
\draw [line width=2pt, short] (-7.5,5.75) -- (-10,2);
\draw [ fill={rgb,255:red,0; green,0; blue,0} , line width=0.2pt ] (-10,2) circle (0.25cm);
\node [font=\Huge] at (-9.4,2) {g};
\draw [ fill={rgb,255:red,0; green,0; blue,0} , line width=0.2pt ] (-7.5,5.75) circle (0.25cm);
\node [font=\Huge] at (-6.9,5.75) {d};
\draw [ fill={rgb,255:red,0; green,0; blue,0} , line width=0.2pt ] (-3.75,5.75) circle (0.25cm);
\node [font=\Huge] at (-3.15,5.75) {e};
\draw [ fill={rgb,255:red,0; green,0; blue,0} , line width=0.2pt ] (2.5,5.75) circle (0.25cm);
\node [font=\Huge] at (3.1,5.75) {f};
\draw [ fill={rgb,255:red,0; green,0; blue,0} , line width=0.2pt ] (0,9.5) circle (0.25cm);
\node [font=\Huge] at (0.6,9.5) {c};
\draw [ fill={rgb,255:red,0; green,0; blue,0} , line width=0.2pt ] (0,2) circle (0.25cm);
\node [font=\Huge] at (0.6,2) {i};
\draw [ fill={rgb,255:red,0; green,0; blue,0} , line width=0.2pt ] (5,2) circle (0.25cm);
\node [font=\Huge] at (5.6,2) {j};

\draw [dashed] (-12,9.5) -- (-5.5,9.5);
\draw [dashed] (-12,5.75) -- (-7.5,5.75);
\node [font=\Huge] at (-11,11.75) {tier 1};
\node [font=\Huge] at (-11,7.75) {tier 2};
\node [font=\Huge] at (-11,3.75) {tier 3};

\draw [ fill={rgb,255:red,0; green,0; blue,0} ] (17.5,13.25) circle (0.25cm);
\node [font=\Huge] at (18.1,13.25) {a};
\draw [ fill={rgb,255:red,0; green,0; blue,0} ] (17.5,11) circle (0.25cm);
\node [font=\Huge] at (18.5,11) {bc};
\draw [line width=2pt, short] (17.5,13.25) -- (17.5,11);
\draw [ fill={rgb,255:red,0; green,0; blue,0} , line width=0.2pt ] (15,9.5) circle (0.25cm);
\node [font=\Huge] at (15.6,9.5) {b};
\draw [line width=2pt, short] (17.5,11) -- (15,9.5);
\draw [line width=2pt, short] (17.5,11) -- (20,9.5);

\draw [line width=2pt, short] (15,9.5) -- (15,5.75);
\draw [ fill={rgb,255:red,0; green,0; blue,0} ] (15,5.75) circle (0.25cm);
\node [font=\Huge] at (15.8,5.75) {de};
\draw [line width=2pt, short] (15,5.75) -- (15,3.5);
\draw [ fill={rgb,255:red,0; green,0; blue,0} ] (15,3.5) circle (0.25cm);
\node [font=\Huge] at (16,3.5) {gh};
\draw [line width=2pt, short] (15,3.5) -- (12.5,2);
\draw [line width=2pt, short] (15,3.5) -- (17.5,2);
\draw [ fill={rgb,255:red,0; green,0; blue,0} ] (12.5,2) circle (0.25cm);
\node [font=\Huge] at (13.2,2) {g};
\draw [ fill={rgb,255:red,0; green,0; blue,0} ] (17.5,2) circle (0.25cm);
\node [font=\Huge] at (18.1,2) {h};

\draw [line width=2pt, short] (20,9.5) -- (22.5,5.75);

\draw [line width=2pt, short] (22.5,3.5) -- (20,2);
\draw [line width=2pt, short] (22.5,3.5) -- (25,2);
\draw [line width=2pt, short] (22.5,5.75) -- (22.5,3.5);
\draw [ fill={rgb,255:red,0; green,0; blue,0} ] (22.5,3.5) circle (0.25cm);
\node [font=\Huge] at (23.3,3.5) {ij};
\draw [ fill={rgb,255:red,0; green,0; blue,0} , line width=0.2pt ] (22.5,5.75) circle (0.25cm);
\node [font=\Huge] at (23.1,5.75) {f};
\draw [ fill={rgb,255:red,0; green,0; blue,0} , line width=0.2pt ] (20,9.5) circle (0.25cm);
\node [font=\Huge] at (20.6,9.5) {c};
\draw [ fill={rgb,255:red,0; green,0; blue,0} , line width=0.2pt ] (20,2) circle (0.25cm);
\node [font=\Huge] at (20.6,2) {i};
\draw [ fill={rgb,255:red,0; green,0; blue,0} , line width=0.2pt ] (25,2) circle (0.25cm);
\node [font=\Huge] at (25.6,2) {j};

\draw [dashed] (8,9.5) -- (15,9.5);
\draw [dashed] (8,5.75) -- (15,5.75);
\draw [dashed] (8,11) -- (17.5,11);
\draw [dashed] (8,3.5) -- (15,3.5);
\node [font=\Huge] at (9,12) {tier 1};
\node [font=\Huge] at (9,10.2) {tier 1'};
\node [font=\Huge] at (9,7.75) {tier 2};
\node [font=\Huge] at (9,4.5) {tier 3};
\node [font=\Huge] at (9,2.5) {tier 3'};

\end{circuitikz}
}%
\caption{For the tiered $X$-tree $T$ shown on the left, the tree $\pinch(T)$ is shown on the right. There are three tiers on $T$, and five tiers on $\pinch(T)$.}
\label{fig: pinched tree}
\end{figure}

\begin{figure}[]
\centering
\resizebox{.9\textwidth}{!}{%
\begin{circuitikz}
\tikzstyle{every node}=[font=\LARGE]

\draw [ fill={rgb,255:red,0; green,0; blue,0} ] (-2.5,13.25) circle (0.25cm) ;
\node [font=\Huge] at (-1.9,13.25) {a};
\draw [ fill={rgb,255:red,0; green,0; blue,0} , line width=0.2pt ] (-5,9.5) circle (0.25cm);
\node [font=\Huge] at (-4.4,9.5) {b};
\draw [line width=2pt, short] (-2.5,13.25) -- (-5,9.5);
\draw [line width=2pt, short] (-2.5,13.25) -- (0,9.5);

\draw [line width=2pt, short] (-5,9.5) -- (-7.5,5.75);
\draw [line width=2pt, short] (-5,9.5) -- (-3.75,5.75);
\draw [line width=2pt, short] (0,9.5) -- (2.5,5.75);

\draw [ fill={rgb,255:red,0; green,0; blue,0} , line width=0.2pt ] (-7.5,5.75) circle (0.25cm);
\node [font=\Huge] at (-6.9,5.75) {d};
\draw [ fill={rgb,255:red,0; green,0; blue,0} , line width=0.2pt ] (-3.75,5.75) circle (0.25cm);
\node [font=\Huge] at (-3.15,5.75) {e};
\draw [ fill={rgb,255:red,0; green,0; blue,0} , line width=0.2pt ] (2.5,5.75) circle (0.25cm);
\node [font=\Huge] at (3.1,5.75) {f};
\draw [ fill={rgb,255:red,0; green,0; blue,0} , line width=0.2pt ] (0,9.5) circle (0.25cm);
\node [font=\Huge] at (0.6,9.5) {c};

\draw [dashed] (-12,9.5) -- (-5.5,9.5);

\node [font=\Huge] at (-11,11.75) {tier 1};
\node [font=\Huge] at (-11,7.75) {tier 2};

\draw [ fill={rgb,255:red,0; green,0; blue,0} ] (17.5,13.25) circle (0.25cm);
\node [font=\Huge] at (18.1,13.25) {a};
\draw [ fill={rgb,255:red,0; green,0; blue,0} ] (17.5,11) circle (0.25cm);
\node [font=\Huge] at (18.5,11) {bc};
\draw [line width=2pt, short] (17.5,13.25) -- (17.5,11);
\draw [ fill={rgb,255:red,0; green,0; blue,0} , line width=0.2pt ] (15,9.5) circle (0.25cm);
\node [font=\Huge] at (15.6,9.5) {b};
\draw [line width=2pt, short] (17.5,11) -- (15,9.5);
\draw [line width=2pt, short] (17.5,11) -- (20,9.5);

\draw [line width=2pt, short] (15,9.5) -- (15,7.25);
\draw [line width=2pt, short] (15,7.25) -- (12.5,5.75);
\draw [line width=2pt, short] (15,7.25) -- (17.5,5.75);
\draw [ fill={rgb,255:red,0; green,0; blue,0} ] (15,7.25) circle (0.25cm);
\draw [ fill={rgb,255:red,0; green,0; blue,0} ] (12.5,5.75) circle (0.25cm);
\draw [ fill={rgb,255:red,0; green,0; blue,0} ] (17.5,5.75) circle (0.25cm);
\node [font=\Huge] at (15.8,7.25) {de};
\node [font=\Huge] at (13.3,5.75) {d};
\node [font=\Huge] at (18.3,5.75) {e};

\node [font=\Huge] at (23.1,7.25) {f};
\node [font=\Huge] at (23.1,5.75) {f'};

\draw [line width=2pt, short] (20,9.5) -- (22.5,7.25);
\draw [line width=2pt, short] (22.5,7.25) -- (22.5,5.75);
\draw [ fill={rgb,255:red,0; green,0; blue,0} , line width=0.2pt ] (22.5,7.25) circle (0.25cm);
\draw [ fill={rgb,255:red,0; green,0; blue,0} , line width=0.2pt ] (22.5,5.75) circle (0.25cm);

\draw [ fill={rgb,255:red,0; green,0; blue,0} , line width=0.2pt ] (20,9.5) circle (0.25cm);
\node [font=\Huge] at (20.6,9.5) {c};

\draw [dashed] (8,9.5) -- (15,9.5);
\draw [dashed] (8,11) -- (17.5,11);
\draw [dashed] (8,7.25) -- (15,7.25);

\node [font=\Huge] at (9,12) {tier 1};
\node [font=\Huge] at (9,10.2) {tier 1'};
\node [font=\Huge] at (9,8.2) {tier 2};
\node [font=\Huge] at (9,6.5) {tier 2'};

\end{circuitikz}
}%
\caption{For the tree $T$ shown on the left, the tree $\pinch(T)$ is shown on the right. There are two tiers on $T$, and three tiers on $\pinch(T)$.}
\label{fig: pinched tree small}
\end{figure}
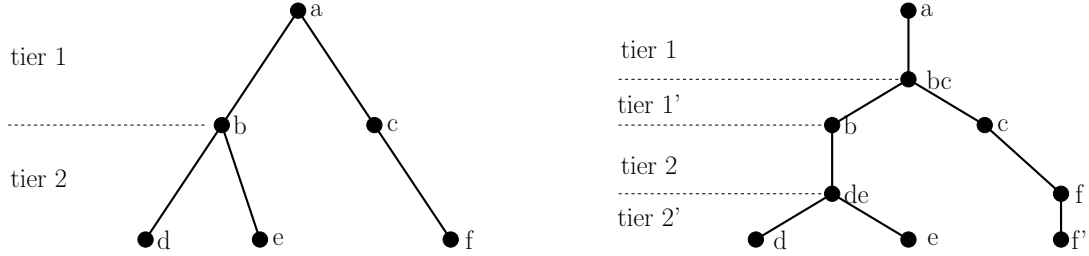

\begin{ex}
    Consider the tree $T$ shown in Figure \ref{fig: pinched tree}. 
    There are four pairs of sibling edges in $T$:
    \[
a \rightarrow b, a \rightarrow c, \quad  b \rightarrow d, b  \rightarrow e,  
\quad  d \rightarrow g, e \rightarrow h, \quad \mbox{ and } \quad 
f \rightarrow i, f \rightarrow j
    \]
 In $\pinch(T)$, the structure corresponding with $a\rightarrow b$ and $a \rightarrow c$ consists of vertices $a, bc, b$, and $c$ and edges $a\rightarrow bc$, $bc \rightarrow b$, and $bc \rightarrow c$. 
 Where $T$ has the single tier $1$, $\pinch(T)$ has two tiers, 
 $1$ and $1'$, with the vertex $bc$ on the border between the two. 
 The sibling edges $b\rightarrow d$ and $b\rightarrow e$ in $T$ correspond 
 to the edge $b \rightarrow de$ in $\pinch(T)$. 
 Since $d$ and $e$ are not terminal sibling vertices in $T$, 
 there are no vertices $d$ and $e$ in $\pinch(T)$ and therefore 
 tier $2$ remains the same in both $\pinch(T)$ and $T$. 
 The edges $d \rightarrow g$ and $e \rightarrow h$ in $T$ correspond to the structure in $\pinch(T)$ with vertices $de$, $gh$, $g$, and $h$, and edges $de \rightarrow gh$, $gh \rightarrow g$, and $gh \rightarrow h$.
 The sibling edges $f \rightarrow i$ and $f\rightarrow j$ in $T$ correspond to the structure in $\pinch(T)$ with vertices $f$, $ij$, $i$, and $j$ and edges $f \rightarrow ij$, $ij \rightarrow i$, and $ij \rightarrow j$. Because $g,h,i$, and $j$ are terminal vertices in $T$, they also appear in $\pinch(T)$. 
 Additionally in $\pinch(T)$, there is the added tier $3'$.
\end{ex}

We will associate a random variable to each vertex in $\pinch(T)$.  
Each vertex $ab$ which resulted from combining vertices $a$ and $b$ in $T$ 
will have a random variable $Y_{ab}$ with state space $S^2$.
For each vertex $a \in V'$, we will have a random variable $Y_a$ with
state space $S$. The graphical model on $\pinch(T)$ is the set of all probability distributions which satisfy the recursive factorization property on $\pinch(T)$. We will now show that the graphical model on $\pinch(T)$ actually contains $M(T,G^{loc})$.

\begin{lemma}
    The set of probability distributions of the model $M(T,G^{loc})$ is contained in the graphical model associated to $\pinch(T)$. 
    \label{model_is_graphical_model}
\end{lemma}

\begin{proof}  

First, we will group variables by tier similarly to how we grouped variables in $T$. Let each vertex be labeled $Y_i$ such that the root is $Y_0$. Then the variables at the conclusion of each tier are grouped into a new variable $X_l$, indexed by the number of tiers with $Y_0=X_0$. Each variable $X_l$ is a vector of $Y_i$ variables whose length is the number of edges of $\pinch(T)$ in tier $l$. If $X_l$ contains $m$ many $Y$ variables with state space $S^2$ and $n$ many $Y$ variables with state space $S$, then the state of $X_l$ is a tuple of size $2m+n$. 

For each edge $e$ in $\pinch(T)$ which is directed toward a vertex having state space $S^2$, the rate matrix will be the same as the rate matrix of the two sibling edges in $T$ which were pinched together to form $e$. For each edge $e$ in $\pinch(T)$ directed toward a vertex with state space $S$, the rate matrix will be exactly the same as the rate matrix for $e$ in $T$.

Suppose there is a tier $l$ in $T$ which has been split 
into two tiers in $\pinch(T)$, which we will call 
$l$ and $l'$. In the first tier $l$, the transition 
matrices will match the transition matrices of the 
equivalent edges (or the equivalent sibling edges) so that $l$ in $\pinch(T)$ has the same length in time as $l$ in $T$. The tier $l'$ will behave as if there is no time passing. The transition matrices on all edges between degree-two vertices will be identity matrices. The two edges with a degree-three parent will have a variation on an identity matrix which reduces the dimension by half. The left edge will have transition matrix $I\otimes \textbf{1}$, and the right edge will have transition matrix $\textbf{1}\otimes I$, where $\textbf{1}$ is a column of ones in $\rr^r$, and $I$ is the identity matrix with size $r\times r$.
The interpretation of these transition matrices are that they extract from the random variable $Y_{ab}  = (X_a, X_b)$ the two random variables $X_a$ and $X_b$, respectively.

Now we follow the same process as in Section  \ref{sec:Markov}, 
where we obtain large transition matrices for each tier 
by summing together the transition matrices for the individual edges. 
Then we arrive at the distribution of the state of the leaves' variable $X_l$ by multiplying the transition matrices for each tier, along with the root distribution:
\[
\pp(X_l=x_l)= \sum_{x_0\in S}\cdots \sum_{x_{l-1} \in S^{n}} \pp(X_0=x_0)\prod_{1\leq i \leq l} \pp(X_i=x_i |X_{i-1} = x_{i-1})   
\]
where $n$ is the number of leaves.

To summarize, the set of probability distributions of the model $M(T,G^{loc})$ is a subset of distributions in the graphical model associated to $\pinch(T)$ with the following properties:
\begin{itemize}
\item For any two sibling edges $a\rightarrow b$ and $a\rightarrow c$ in $T$, the transition matrix for $a\rightarrow bc$ in $\pinch(T)$ is the same as the transition matrix for the two edges $a\rightarrow b$ and $a\rightarrow c$ in $T$
\item For any two sibling edges $a\rightarrow c$ and $b\rightarrow d$ in $T$, the transition matrix for $ab\rightarrow cd$ in $\pinch(T)$ is the same as the transition matrix for the two edges $a\rightarrow c$ and $b\rightarrow d$ in $T$
\item For vertices $uv$, $u$, and $v$, with state spaces $S^2, S,$ and $S$ respectively, if $uv$ has state $(\mu,\nu)$, then $u$ takes state $\mu$ and $v$ takes state $\nu$.
\end{itemize}
    From this we can see that for all vertices $v$ of $\pinch(T)$, $Y_v \indep Y_{\nd(v)\backslash \pa(v)}|Y_{\pa(v)}$.
\end{proof}

To elaborate on the assignment of rate matrices, we provide the following examples. Most of the edges of $\pinch(T)$ have very similar matrices to those in $T$, with the exception of sibling edges of $T$ that have been pinched together. Although there is a reduction in the number of edges, the combined rate matrix for the pair of sibling edges is the same as the rate matrix of the pinched together edge in $\pinch(T)$.

\begin{ex}
    Let $T$ be a tiered X-tree and let $a\rightarrow b$ and $a\rightarrow c$ be sibling edges. This pair of edges has a rate matrix $Q_{bc}$ as described in Definition \ref{def: construct rate matrix}. In $\pinch(T)$ there is an edge $a\rightarrow bc$  with rate matrix $Q_{bc}$ (the same rate matrix as in $T$). Given a distribution $\pi_a$ on vertex $a$, the distribution on $bc$ given $a$ is $\pi_{a}\trunc(P_{bc})$ where $P_{bc}=\exp(Q_{bc})$.

    Similarly, let $a\rightarrow c$ and $b\rightarrow d$ be sibling edges in $T$ with rate matrix $Q_{cd}$. Then there is an edge $ab\rightarrow cd$ in $\pinch(T)$ with rate matrix $Q_{cd}$. Given a distribution $\pi_{ab}$ on the vertex $ab$, the distribution on $cd$ is $\pi_{ab}P_{cd}$ where $P_{cd}=\exp(Q_{cd})$.
 \end{ex}

The added tiers of $\pinch(T)$, consisting of edges with identity matrices and modified identity matrices differ from the tiers in $T$. These tiers behave as if no time passes across them, 
and they only exist as a tool to split the random variables with state space $S^2$ into two separate random variables with state space $S$. These edges directed away from the variables with state space $S^2$ have matrices which select the first or second state of the ordered pair. The other edges in this tier have plain identity matrices.

\begin{ex}
For the tree pictured in Figure \ref{fig: pinched tree}, with $S=\{0,1\}$, the edge $gh\rightarrow g$ would have transition matrix
\[
\begin{bmatrix}
    1&0\\0&1
\end{bmatrix} \otimes \begin{bmatrix}
    1\\1
\end{bmatrix} = \begin{bmatrix}
    1&0\\1&0\\0&1\\0&1
\end{bmatrix}
\]
and the edge $gh\rightarrow h$ would have transition matrix 
\[
\begin{bmatrix}
    1\\1
\end{bmatrix} \otimes \begin{bmatrix}
    1&0\\0&1
\end{bmatrix} = \begin{bmatrix}
    1&0\\0&1\\1&0\\0&1
\end{bmatrix}.
\]

The conditional probabilities on the random variables associated to $g$, $h$, and $gh$ are
\[
\pp (X_g=x_g | X_{gh}=x_{gh})=\pp(X_g=x_g|X_{gh}=(\chi_g,\chi_h)) = \begin{cases}
    1 & x_g=\chi_g \\
    0 & x_g \neq \chi_g
\end{cases},
\]
\[
\pp (X_h=x_h | X_{gh}=x_{gh})=\pp(X_h=x_h|X_{gh}=(\chi_g,\chi_h)) = \begin{cases}
    1 & x_h=\chi_h \\
    0 & x_h \neq \chi_h
\end{cases}.
\]

In Figure \ref{fig: pinched tree small}, with $S=\{0,1\}$, the edge $f\rightarrow f'$ would have the identity for a transition matrix.

\end{ex}

The end result is that each tier of $T$ has the same large transition matrix as the corresponding tier of $\pinch(T)$.

\begin{ex}
    Consider the tree $T$ shown in Figure \ref{fig: pinched tree} with its associated $\pinch(T)$.  The transition matrices for tiers 1, 2, and 3 will be the same for both tiered trees. However, in $\pinch(T)$, the tiers 1' and 3' will have transition matrices constructed from identity matrices as described above.
\end{ex}

 Now that we are equipped with a graphical model representation of 
 the local lateral gene transfer model, to prove identifiability
 results, we need a way to determine 
 the structure of the tree based on probability distributions produced by the model
 on that tree. 
 In phylogenetic models, we only observe data at 
 the leaves of the tree while the rest of the variables are hidden from view. 
 We can take the leaf data and put it into a matrix called a flattening matrix. 
 Then conditional independence relations will become visible 
 in the rank of the flattening matrix, as we will describe in Lemma \ref{lemma: rank of flat}.   Low rank of the flattening matrices will allow us to identify 
 splits in the tree $T$.  
For this to work, we need
  that the splits of $T$ and $\pinch(T)$ are identical, so that information on the splits of $\pinch(T)$ will be relevant.

\begin{lemma}
\label{splits_of_pinch}
  Let $T$ be a tiered X-tree.  Then $T$ and $\pinch(T)$ have the same set of splits. 
  That is, $\Sigma(T)  =  \Sigma(\pinch(T))$.
\end{lemma}

\begin{proof}
    Let $A|B \in \Sigma(T)$, represented by an edge $e$. If $e$ is not a sibling edge in $T$, then $e$ will also be an edge of $\pinch(T)$ representing $A|B$. Suppose $e$ is a terminal edge in $T$.  Also suppose $e$ is incident with vertices $a$ and $c$ such that $a\rightarrow c$ and has a sibling edge incident with vertices $b$ and $d$ such that $b\rightarrow d$. Then in $\pinch(T)$ there are vertices $ab, cd, c,$ and $d$ with edges $ab\rightarrow cd$, $cd\rightarrow c$, and $cd\rightarrow d$. Then the edge $cd\rightarrow c$ represents the split $A|B$ in $\pinch(T)$. 

    Now let $A|B$ be a split in  $\pinch(T)$ represented by an 
    edge $e=u\rightarrow v$. If $v$ is not in $\sib(T)$, 
    then the edge in $T$ incident with and directed toward $v$ 
    represents $A|B$ in $T$. If $v$ is in $\sib(T)$, 
    then there are ancestors $a$ and $b$ of $v$ such that $a\rightarrow b$ is an edge of $\pinch(T)$, neither $a$ nor $b$ is in $\sib(T)$, 
    and the path between $v$ and $a$ passes through only degree-two vertices (it is possible that $u=b$). Then, the edge $a\rightarrow b$ represents the split $A|B$ in $\pinch(T)$, and there is a corresponding edge $a\rightarrow b$ in $T$ which also represents $A|B$. 
\end{proof}

\begin{defn}
Let $A|B$ be an $X$-split.  Let $P= (p(k):  k \in S^{\#X})$ be a joint probability
distribution.  Let $A = \{a_1, \ldots, a_n\}$ and
$B = \{b_1, \ldots, b_m \}$.  Then the \textit{flattening matrix} 
$\Flat_{A|B}(P)$ is the $r^{\#A}\times r^{\#B}$ matrix 
where the rows are by tuples $k_A = (k_{a_1}, \ldots, k_{a_n}) \in S^{\#A}$
and the columns are indexed by tuples $l_B = (l_{b_1}, \ldots, l_{b_m}) \in S^{\#B}$
and where the 
$(k_A,l_B)$-entry of $\Flat_{A|B}(P)$  is 
$p(k_A, l_B)$.
\end{defn}

\begin{ex}
Consider a four leaf tree with taxa $\{1,2,3,4\}$, a split $12|34$, and where $S=\{ 0,1 \}$. The flattening matrix corresponding to this split is the following matrix, where each entry shows the probability of that sequence of states being displayed by the taxa.
\[ \mathrm{Flat}_{12|34}(P)=
    \begin{blockarray}{ccccc}
        & \cdot \cdot  00 & \cdot \cdot  01 &  \cdot \cdot  10 & \cdot \cdot  11 \\
      \begin{block}{c[cccc]}
        00 \cdot \cdot & p(0000) & p(0001) & p(0010) & p(0011) \\
        01 \cdot \cdot  & p(0100) & p(0101) & p(0110) & p(0111) \\
        10 \cdot \cdot  & p(1000) & p(1001) & p(1010) & p(1011) \\
        11 \cdot \cdot  & p(1100) & p(1101) & p(1110) & p(1111) \\
      \end{block}
    \end{blockarray}
\]
\end{ex}


Flattening matrices will help us compare various splits of taxa on our tree, but there is some nuance in the difference between splits of a standard tree and splits of a pinched tree. When examining splits of a tree, sometimes it is useful to subdivide the edge corresponding to a certain split and re-root the tree on the new vertex. Typically, in a phylogenetic tree each random variable has the same state space, so the new root will be similar to all other vertices. In $\pinch(T)$, some vertices have state space $S$ and others have state space $S^2$. However, when we choose the edge which corresponds to a split as in the proof above, we can subdivide it to obtain a vertex with state space $S$. This notion is critical in the proof of the next lemma.

\begin{lemma}
   \label{lemma: rank of flat}
    Let $T$ be a tree, and let $M(T,G^{loc})$ be the 
    lateral gene transfer model on $T$ with $\#S = r$. 
    Then for any $A|B\in \Sigma(T)$ and any $P\in M(T,G^{loc})$
\[\rank(\Flat_{A|B}(P))\leq r.\]
Furthermore, if $A|B\notin \Sigma(T)$, then for generic $P\in M(T,G^{loc})$
\[\rank(\Flat_{A|B}(P))\geq r^2.\]
\end{lemma}
\begin{proof}
    The proof will closely follow the proofs in \cite{Tree_Symmetries}, \cite{allman2006phylogeneticidealsvarietiesgeneral}, and \cite[Proposition 15.4.4]{Algebraic_Statistics}. 

    In this proof we assume that $|A|,|B|\geq 2$. First we will subdivide the edge of $\pinch(T)$ as described in Lemma \ref{splits_of_pinch} and make the new degree two vertex the root of the tree. We can treat each side of this split as its own tree in which large tier-encompassing rate matrices are formed. Then the joint distribution of states on the leaves is 
    \[
    P_{i_1 \ldots i_m}=\sum_{j=1}^r \pi_j F(i_A, j)G(j,i_B)
    \]
    where $F(i_A,j)$ is the function obtained by summing out over all the states of the  hidden variables between the root node and the set of leaves $A$, and $G(j,i_B)$ is the function obtained by summing out over all the hidden variables between the root node and the set of leaves $B$. Let $F$ denote the resulting $r^{|A|}\times r$ matrix which is the flattening of this tensor and similarly let $G$ denote the resulting $r\times r^{|B|}$ matrix. Then the joint distribution on the leaves can be expressed as 
    \[
    \Flat_{A|B}(P)=F \diag (\pi)G
    \]
    which shows that $\rank(\Flat_{A|B}(P))\leq r$.
   
    For the second part of the theorem, note that it suffices to find a single $P\in M(T,G^{loc})$ such that $\rank(\Flat_{A|B}(P))\geq r^2$. This is because the set of matrices with $\rank(\Flat_{A|B}(P))< r^2$ is the variety $\mathcal{V}_{A|B} = \{ P\in \rr^{r^{|X|}}| \det(\Flat_{A|B}(P))=0\} $, and the Zariski closure of the model $M(T,G^{loc})$ is an irreducible variety. Therefore, if there is a single point $P\in M(T,G^{loc})$ such that $\det(\Flat_{A|B}(P))\neq 0$, $M(T,G^{loc})$ is not contained in $\mathcal{V}$, and the intersection of $\mathcal{V}$ and $M(T,G^{loc})$ is a proper subvariety of $M(T,G^{loc})$. This intersection must then be of measure zero within $M(T,G^{loc})$, and $\Flat_{A|B}(P)$ generically has full rank. 
    
    So suppose that $A|B \notin \Sigma(T)$. This means that there exists another split $C|D\in \Sigma(T)$ such that $A|B$ and $C|D$ are incompatible, and the edge corresponding to $C|D$ can be selected by Lemma \ref{splits_of_pinch}. We construct the distribution $P$ as follows. 
    Set all lambda parameters to zero; this way the model reduces to the case with no lateral gene transfer. Set all rate matrices except the one corresponding to the edge $C|D$ to the zero matrix, and set the rate matrix for the edge corresponding to $C|D$ to an arbitrary $r\times r$ rate matrix $\alpha$. Additionally, let each tier have length $1$. The rate matrices are then combined using the Kronecker sum, appropriately truncated, and exponentiated as in Section \ref{sec:Markov}. 
    
    The resulting joint distribution $P$ has the following form:
    \[
    P_{i_1...i_m}=\begin{cases}
        \exp{(\alpha)}_{jk} & \text{if } i_c=j \text{ and } i_d=k \text{ for all } c\in C,d\in D, \\ 0 & \text{otherwise.}
    \end{cases}
    \]
    Since $A|B$ and $C|D$ are not compatible this means that none of $A \cap C$, $A\cap D$, $B\cap C$, or $B\cap D$ are empty. Now form the sub-matrix $M$ of $\Flat_{A|B}(P)$ whose row indices consist of all tuples $i_A$ such that $i_{A\cap C}$ and $i_{A\cap D}$ are constant, and whose column indices consist of all tuples $i_B$ such that $i_{B\cap C}$ and $i_{B\cap D}$ are constant. The matrix $M$ will be a $r^2\times r^2$ matrix which has exactly one nonzero entry in each row and column. Such a matrix has rank $r^2$ which implies that $\rank(\Flat_{A|B}(P))\geq r^2$.
\end{proof}

In this way, we examine the ranks of flattening matrices on the graph $\pinch(T)$, while still gaining information on $T$, since $T$ and $\pinch(T)$ have the same set of splits. Then we can begin to distinguish between trees, by comparing the valid or invalid splits from each tree. This leads to the main result.


\begin{thm}
    Let $T$ and $T'$ be binary tiered X-trees with differing unrooted topologies. Then under  the Local Gene Transfer Model these trees are generically distinguishable.
\end{thm}
\begin{proof}

Let $T$ and $T'$ be tiered $X$-trees under the Local Gene Transfer Model with differing unrooted topologies. 
Since the unrooted topologies of $T$ and $T'$ are different, there is a split $A|B \in \Sigma(T)
\setminus \Sigma(T')$ and a splits $C|D \in \Sigma(T') \setminus \Sigma(T)$.
We apply Lemma \ref{lemma: rank of flat}.
Let $P \in M(T,G^{loc})$ be any probability distribution.  According to 
 Lemma \ref{lemma: rank of flat} ,  $\rank(\Flat_{A|B}(P)) \leq r $.  In addition, 
if $P$ is generic, then $\rank(\Flat_{C|D}(P)) \geq r^2 $.
 On the other hand, if $P' \in M(T',G^{loc})$ is any probability distribution, then 
 $\rank(\Flat_{C|D}(P')) \leq r $.  In addition, if $P'$ is generic, then 
 $\rank(\Flat_{A|B}(P')) \geq r^2 $.  Since having rank $\leq r$ is a closed condition,
 this shows that 
\[
 \dim(M(T,G^{loc}) \cap  M(T',G^{loc}) )  <  \min  (\dim M(T,G^{loc}), \dim M(T',G^{loc}))
\]
Thus $T$ and $T'$ are generically distinguishable.
\end{proof}


\section{Simulation Studies}  \label{sec:sims}

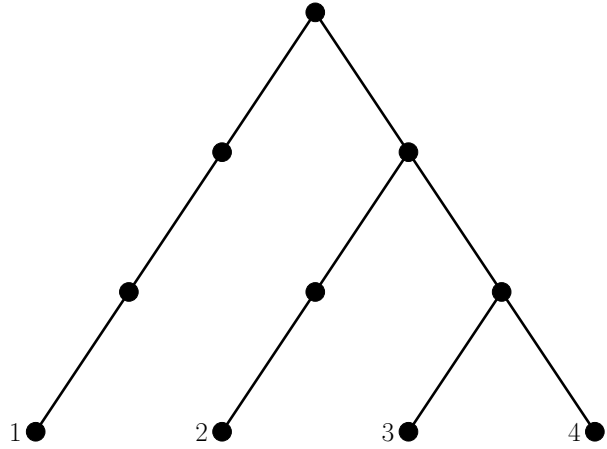
\begin{figure}[]
\centering
\resizebox{.5\textwidth}{!}{%
\begin{circuitikz}
\tikzstyle{every node}=[font=\huge]

\draw [ fill={rgb,255:red,0; green,0; blue,0} ] (-2.5,13.25) circle (0.25cm) node [left=7pt]{};

\draw [line width=2pt, short] (-2.5,13.25) -- (-5,9.5);
\draw [line width=2pt, short] (-2.5,13.25) -- (0,9.5) ;

\draw [ fill={rgb,255:red,0; green,0; blue,0} ] (0,9.5) circle (0.25cm) node [left=7pt]{};
\draw [ fill={rgb,255:red,0; green,0; blue,0} ] (-5,9.5) circle (0.25cm) node [left=7pt]{};
\draw [ fill={rgb,255:red,0; green,0; blue,0} , line width=0.2pt ] (2.5,5.75) circle (0.25cm) node [left=7pt]{};
\draw [ fill={rgb,255:red,0; green,0; blue,0} , line width=0.2pt ] (-2.5,5.75) circle (0.25cm) node [left=7pt]{};
\draw [ fill={rgb,255:red,0; green,0; blue,0} , line width=0.2pt ] (-7.5,5.75) circle (0.25cm) node [left=7pt]{};
\draw [line width=2pt, short] (-5,9.5) -- (-7.5,5.75);
\draw [line width=2pt, short] (0,9.5) -- (-2.5,5.75);

\draw [line width=2pt, short] (-2.5,5.75) -- (-5,2);
\draw [line width=2pt, short] (2.5,5.75) -- (0,2) ;

\node [font=\LARGE] at (-3.25,3.25) {};
\draw [line width=2pt, short] (0,9.5) -- (2.5,5.75);
\draw [line width=2pt, short] (2.5,5.75) -- (5,2);
\draw [line width=2pt, short] (-7.5,5.75) -- (-10,2);
\draw [ fill={rgb,255:red,0; green,0; blue,0} , line width=0.2pt ] (-10,2) circle (0.25cm) node [left=7pt]{$1$};
\draw [ fill={rgb,255:red,0; green,0; blue,0} , line width=0.2pt ] (0,2) circle (0.25cm) node [left=7pt]{$3$};
\draw [ fill={rgb,255:red,0; green,0; blue,0} , line width=0.2pt ] (-5,2) circle (0.25cm) node [left=7pt]{$2$};
\draw [ fill={rgb,255:red,0; green,0; blue,0} , line width=0.2pt ] (5,2) circle (0.25cm) node [left=7pt]{$4$};

\end{circuitikz}
}%
\caption{A rooted caterpillar tree on four leaves displaying split $12|34$.}
\label{fig: caterpiller sim}
\end{figure}

Besides being a tool for proving identifiability results, flattening 
matrices with low rank can be used as a tool to reconstruct a phylogenetic
tree from data.  The standard approach involves using the
singular value decomposition to test if a flattening matrix is
close to low rank  (for example, the SVD Quartets method \cite{Chifman2014-kl}).
In this section, we explore this approach on data simulated from
the Local Gene Transfer model.

 We begin by constructing a probability distribution on a caterpillar tree with split $12|34$ (Figure \ref{fig: caterpiller sim}). We will use a binary model with $S=\{0,1\}$, and with each tier having length $1$. 
We iterate over mutation rate values between $0$ and $1$ so that the rate of change from state $0$ to state $1$ and from state $1$ to state $0$ are always equal. We further iterate over lateral gene transfer rates between $0$ and $0.5$. For each possible mutation/LGT parameter pairing, we generate a probability distribution by following the process outlined in Sections 2 and 3. Due to numerical error, these probability distribution vectors occasionally do not sum to one (with the sum having a difference of less than $10^{-8}$). To correct this issue we normalize the vector, thus preserving the ratios between each probability. From this distribution we generate a sequence of sites of length $1000$ and construct an empirical distribution $\hat{P}$ based on the proportion of each site pattern generated. 

We arrange the simulated distributions into the three possible flattening matrices: 
$\mathrm{Flat}_{12|34}(\hat{P})$, $\mathrm{Flat}_{13|24}(\hat{P})$, and $\mathrm{Flat}_{14|23}(\hat{P})$. 
For each matrix we compute the SVD score \cite{Chifman2014-kl}, which will compute the distance between the matrix and the nearest matrix of rank 2:
\[
SVD(A)=\sqrt{\sigma_3^2+\sigma_4^2},
\]
where $A$ is a matrix, and $\sigma_3$ and $\sigma_4$ are the third and fourth singular values of $A$. 
Thus a small SVD score indicates that the matrix is close to being rank two. 
The flattening matrix with the smallest SVD score is chosen as the inferred split of the original tree. After one hundred repetitions of this process, we record the percent for each mutation-LGT parameter pairing that the correct split was inferred. 
\begin{figure}
    \centering
    \includegraphics[width=0.8\linewidth]{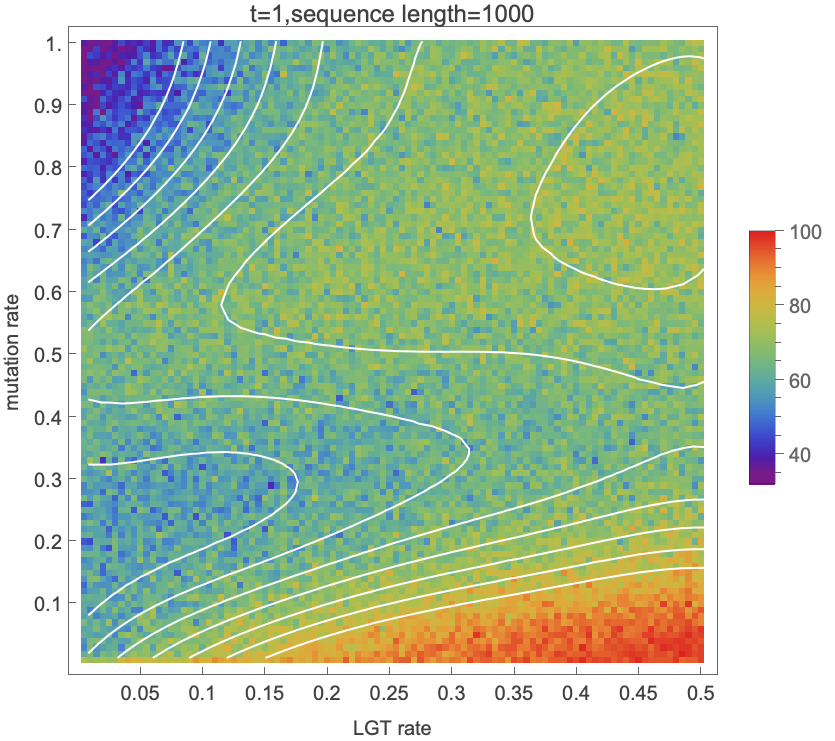}
    
    \caption{Each square of the plot represents 100 trials of inferring the split of the tree in Figure \ref{fig: caterpiller sim}. The percent of time the correct split was predicted is represented by the shade of the square, with violet being a low percentage and red being a higher percentage. Overlaid on the plot are ten isoclines created using Mathematica's model fit function \textit{kernel}.}
    \label{fig:heatmap}
\end{figure}

Figure \ref{fig:heatmap} shows that under the Local Gene Transfer model, an increase in the rate of lateral gene transfer improves the probability of correctly inferring tree structure. This effect is likely because as the rate of lateral gene transfer increase, the divergence of lineages is pushed farther forward in time. Since lateral gene transfer is restricted to sibling pairs, it has a strengthening effect on the relatedness of the two lineages. In the case of Figure \ref{fig: caterpiller sim} there is less time for the sibling lineages 3 and 4 to mutate away from each other. The branch lengths of tier 3 are effectively shortened, decreasing the possibility of mutations, while the length of tier 1 is increased, making mutations more likely. Tier 2 remains a significant length which allows for the split $12|34$ to be solidified.

The mutation rate has a different effect. When the mutation rate is high, with low LGT, there is too much mutation for the correct tree structure to be inferred. In Figure \ref{fig:heatmap}, the upper left corner shows very low rates of success, nearly at a completely random level. Interestingly, there is a band of relatively low success rate where the mutation rate is roughly between $0.2$ to $0.4$. This may be because there is not enough mutation to adequately differentiate between lineages. This band seems to have a success rate of $50\%$, which although low, is still higher than random chance.

\begin{figure}
    \centering
    \subfloat[$\lambda = 0.0$]{\includegraphics[width=0.32 \linewidth]{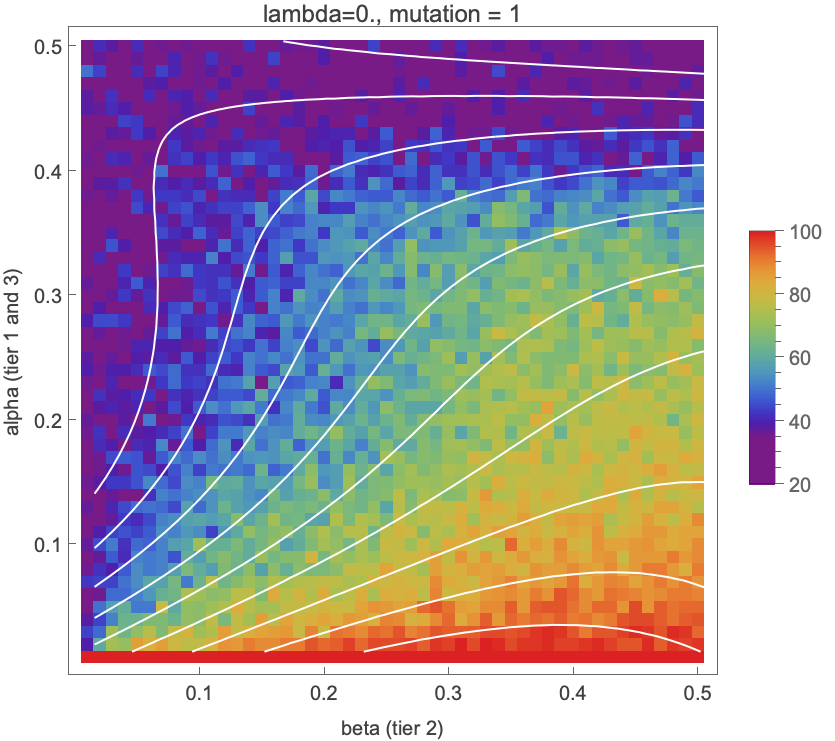}}
    \subfloat[$\lambda = 0.1$]{\includegraphics[width=0.32 \linewidth]{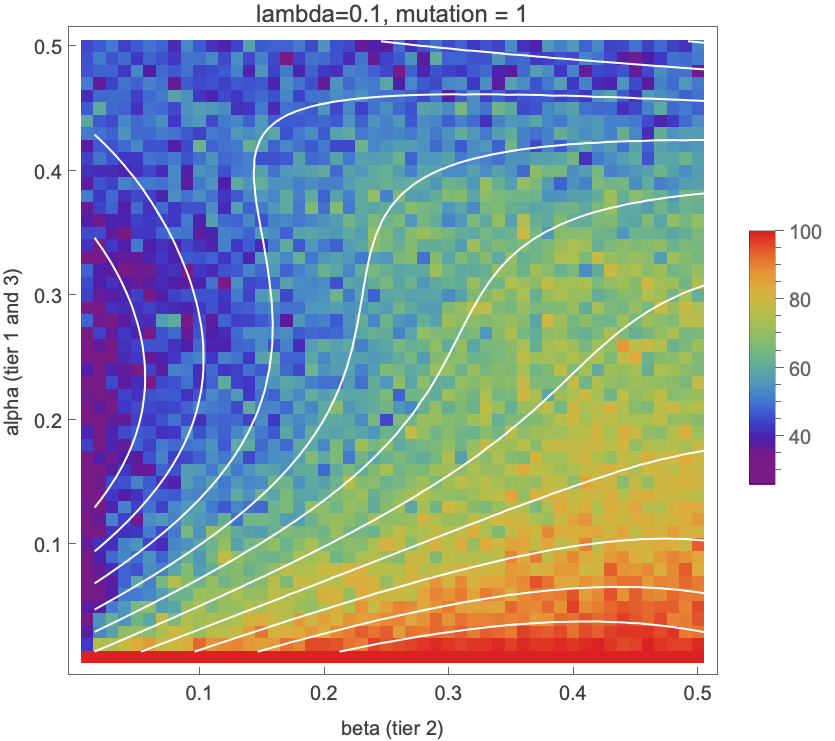}}
    \subfloat[$\lambda = 0.2$]{\includegraphics[width=0.32 \linewidth]{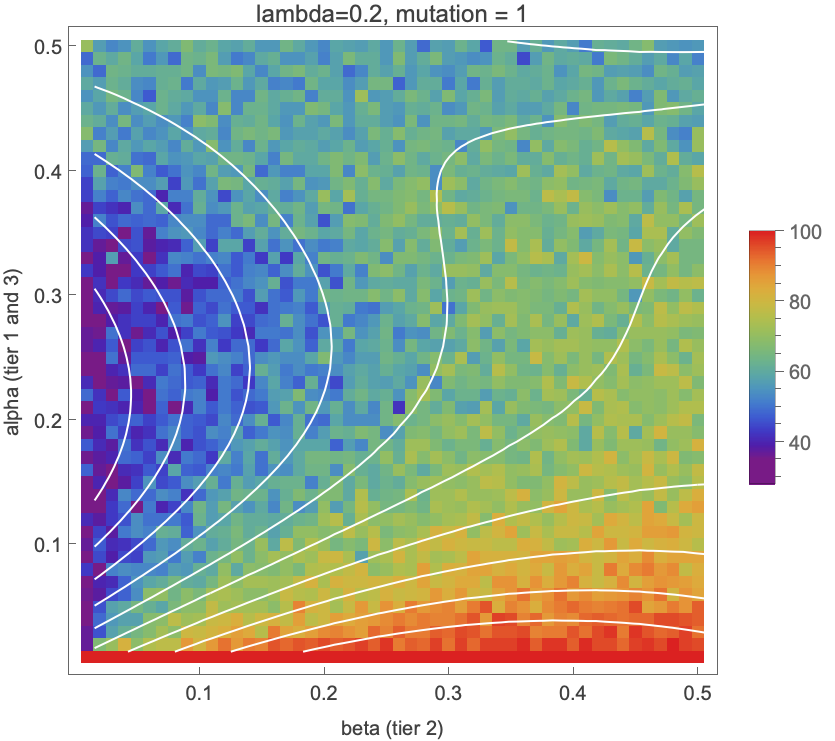}}
    
    \subfloat[$\lambda = 0.3$]{\includegraphics[width=0.32 \linewidth]{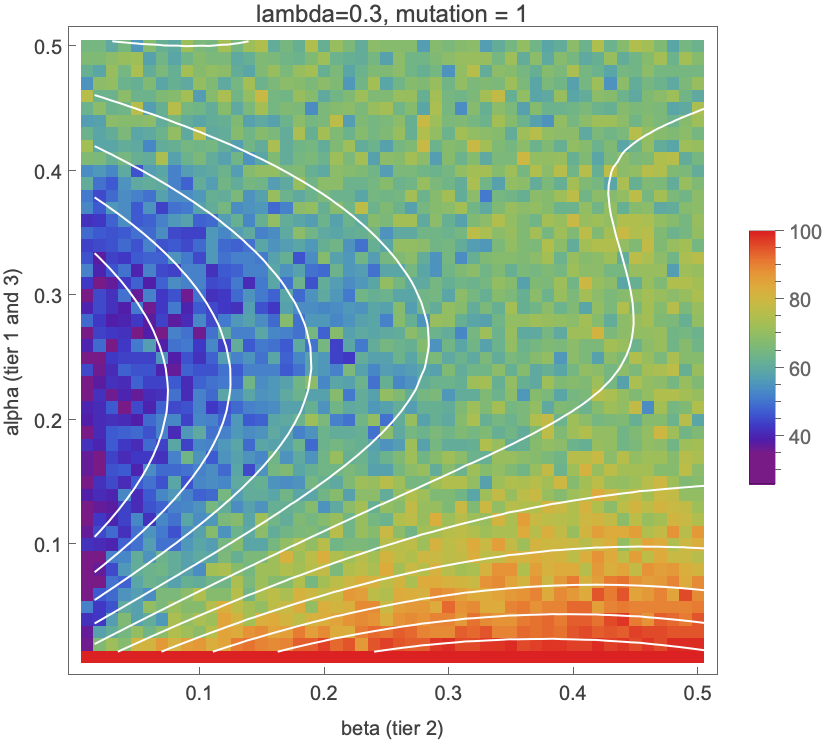}}
    \subfloat[$\lambda = 0.4$]{\includegraphics[width=0.32 \linewidth]{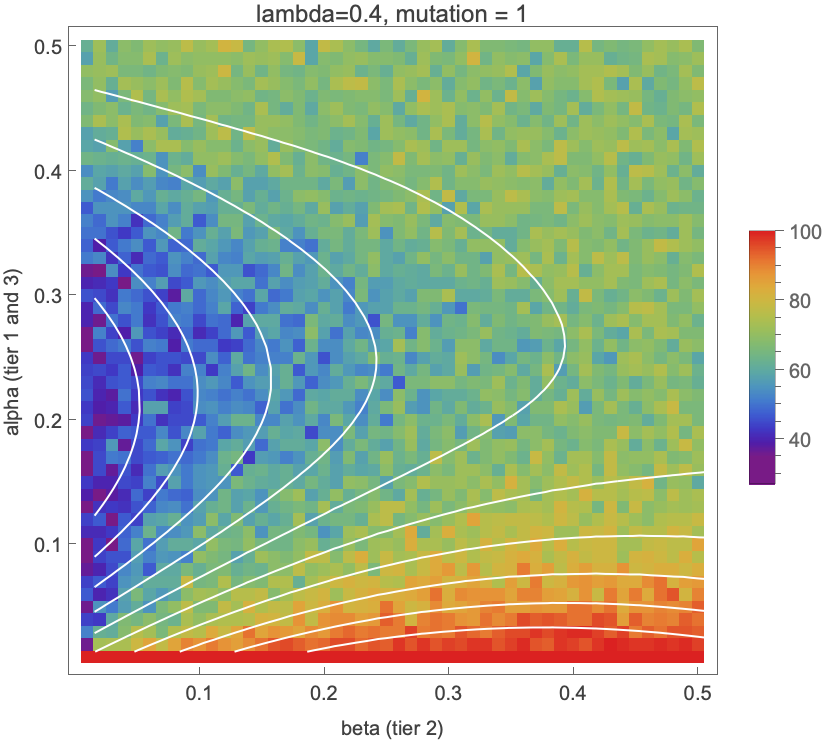}}
    \subfloat[$\lambda = 0.5$]{\includegraphics[width=0.32 \linewidth]{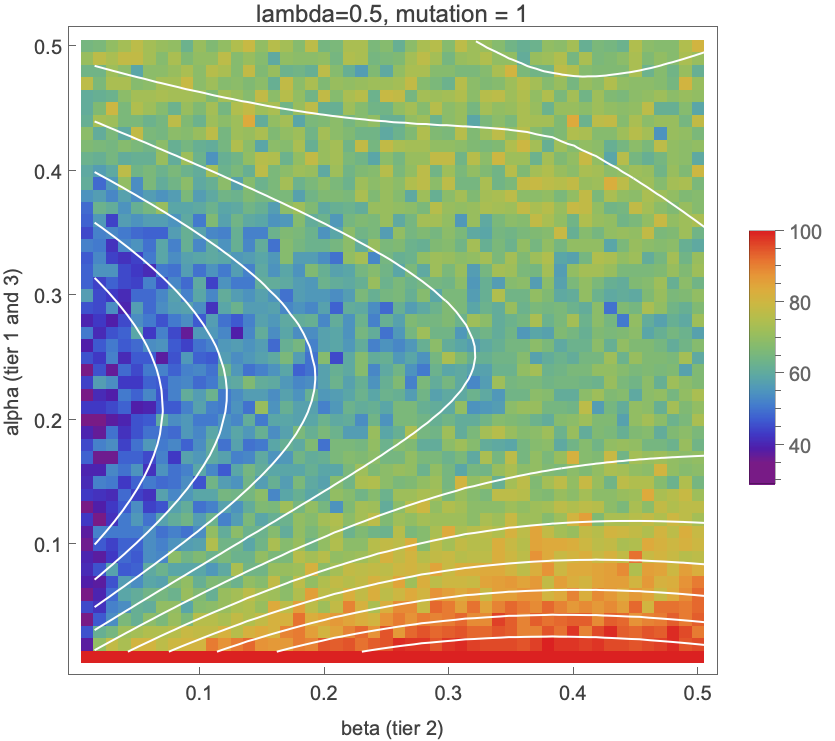}}
    \caption{Each graph charts the success rate of determining the correct species tree with the length of the second tier on the horizontal axis and the length of the first and third tiers on the vertical axis. The lengths range from $0$ to infinity, represented here by $t_{1,3}=-0.5\ln{(-2\alpha+1)}, t_2=-0.5\ln{(-2\beta+1)}$. The mutation rate is set to $1$ while the lgt rate ranges from $0.1$ to $0.5$ across the five plots. In white overlaid on the charts are ten isoclines created using a cubic model fit.}
    \label{fig: tier lengths comparison}
\end{figure}

Figure \ref{fig: tier lengths comparison} plots the success rate of correctly choosing the tree topology while comparing the length of tier 2 with the lengths of tiers one and three. In this simulation, the length of tiers one and three were locked together, both ranging from length zero to infinity. The length of tiers one and three are represented on the vertical axis, and similarly the length of tier 2 ranges from zero to infinity on the horizontal axis. This was accomplished using a change of variables:
\begin{align*}
    \alpha &= 0.5 (1-e^{2t_{13}}), \\
    \beta &= 0.5 (1-e^{2t_{2}})
\end{align*}
where $t_{13}$ is the length of each of tiers 1 and 3, and $t_2$ is the length of tier 2. This chart is plotted for six values of LGT between $0$ and $0.5$. The plots indicate that as tier 2 increases in length, the success rate increases. We believe this is likely 
due to the fact that as tier 2 increases in length, the edge corresponding to split $12|34$ increases in length. Short lengths for tiers 1 and 3 generally strengthen the bond between the pairs of lineages $(1, 2)$, and $(3, 4)$, especially with high rates of lateral gene transfer. However, with moderate lengths of tiers 1 and 3 and a short tier 2, there may be just enough mutation events to undo any synergy within the lineage pairs. Once tiers 1 and 3 become long enough, the number of mutation events causes randomness such that sufficiently large LGT rates are able to overcome the randomness. LGT events effectively shorten the length of the tier. In this way, even as branch lengths tend toward infinity, the lateral gene transfer process shortens the branches to a more manageable length, where inference is possible.

\bibliography{refs}{}
\bibliographystyle{plain}

\end{document}